\documentclass{article}

\usepackage{times}
\usepackage{amssymb,amsmath,amsthm}
\usepackage{mathrsfs}
\usepackage{epsfig}

\newcommand\apriori{{\em a priori}~}
\newcommand\ie{{\em i.e.}~}
\newcommand\eg{{\em e.g.}~}

\newcommand\etc{{\em etc.}}

\def\B{\mathscr B}

\def\C{\mathbb C}

\def\H{\mathcal H}

\def\N{\mathbb N}
\def\O{\mathcal O}

\def\R{\mathbb R}
\def\S{\mathscr S}

\def\Z{\mathbb Z}

\def\12{{\textstyle\frac12}}
\def\<{\left\langle}
\def\>{\right\rangle}
\def\({\left(}
\def\){\right)}
\def\[{\left[}
\def\]{\right]}
\def\dom{\mathcal D}
\def\lone{\mathsf{L}^{\:\!\!1}}

\def\ltwo{\mathsf{L}^{\:\!\!2}}
\def\linf{\mathsf{L}^{\:\!\!\infty}}
\def\cinf{\mathsf{C}^\infty}

\def\e{\mathop{\mathrm{e}}\nolimits}
\def\d{\mathrm{d}}

\def\im{\mathop{\mathsf{Im}}\nolimits}
\def\supp{\mathop{\mathrm{supp}}\nolimits}
\def\Crit{\mathop{\mathsf{Crit}}\nolimits}

\def\im{\mathop{\mathsf{Im}}\nolimits}
\def\re{\mathop{\mathsf{Re}}\nolimits}

\def\id{\mathop{\mathrm{id}}\nolimits}

\def\ddt{\frac{\mathrm{d}}{\mathrm{d} t}}

\def\grad{\mathrm{grad}}
\def\cone{\mathsf{C}^1}

\newtheorem{Theorem}{Theorem}[section]
\newtheorem{Remark}[Theorem]{Remark}
\newtheorem{Lemma}[Theorem]{Lemma}
\newtheorem{Assumption}[Theorem]{Assumption}
\newtheorem{Corollary}[Theorem]{Corollary}
\newtheorem{Proposition}[Theorem]{Proposition}
\newtheorem{Definition}[Theorem]{Definition}

\begin{document}

%--------------------------------------------------------------------------------------
% Title
%--------------------------------------------------------------------------------------

\title{A formula relating localisation observables to the variation of energy in
Hamiltonian dynamics}

\author{A. Gournay$^1$\footnote{Work made while on leave from the Max Planck
Institute f\"ur Mathematik.}~~and R. Tiedra de Aldecoa$^2$\footnote{Supported by the
Fondecyt Grant 1090008 and by the Iniciativa Cientifica Milenio ICM P07-027-F
``Mathematical Theory of Quantum and Classical Magnetic Systems".}}

\date{\small}
\maketitle \vspace{-1cm}

\begin{quote}
\emph{
\begin{itemize}
\item[$^1$] Universit\'e de Neuch\^atel, Rue E.-Argand 11, CH-2000 Neuch\^atel,
Switzerland.
\item[$^2$] Facultad de Matem\'aticas, Pontificia Universidad Cat\'olica de Chile,\\
Av. Vicu\~na Mackenna 4860, Santiago, Chile
\item[] \emph{E-mails:} gournay@mpim-bonn.mpg.de, rtiedra@mat.puc.cl
\end{itemize}
}
\end{quote}

%--------------------------------------------------------------------------------------
% Abstract
%--------------------------------------------------------------------------------------

\begin{abstract}
We consider on a symplectic manifold $M$ with Poisson bracket
$\{\;\!\cdot\;\!,\;\!\cdot\;\!\}$ an Hamiltonian $H$ with complete flow and a family
$\Phi\equiv(\Phi_1,\ldots,\Phi_d)$ of observables satisfying the condition
$\{\{\Phi_j,H\},H\}=0$ for each $j$. Under these assumptions, we prove a new formula
relating the time evolution of localisation observables defined in terms of $\Phi$ to
the variation of energy along classical orbits. The correspondence between this
formula and a formula established recently in the framework of quantum mechanics is
put into evidence.

Among other examples, our theory applies to Stark Hamiltonians, homogeneous
Hamiltonians, purely kinetic Hamiltonians, the repulsive harmonic potential, the
simple pendulum, central force systems, the Poincar\'e ball model, covering
manifolds, the wave equation, the nonlinear Schr\"odinger equation, the Korteweg-de
Vries equation and quantum Hamiltonians defined via expectation values.
\end{abstract}

\textbf{2000 Mathematics Subject Classification:} 37J05, 37K05, 37N05, 70H05, 70S05.

\newpage
\tableofcontents
\newpage

%--------------------------------------------------------------------------------------
\section{Introduction and main results}\label{Intro}
\setcounter{equation}{0}
%--------------------------------------------------------------------------------------

The purpose of the present paper is to put into evidence a new formula in Hamiltonian
dynamics, both simple and general, relating the time evolution of localisation
observables to the variation of energy along classical orbits.

Our result is the following. Let $M$ be a (finite or infinite-dimensional) symplectic
manifold with symplectic $2$-form $\omega$ and Poisson bracket
$\{\;\!\cdot\;\!,\;\!\cdot\;\!\}$. Let $H\in\cinf(M)$ be an Hamiltonian on $M$ with
complete flow $\{\varphi_t\}_{t\in\R}$. Let
$\Phi\equiv(\Phi_1,\ldots,\Phi_d)\in\cinf(M;\R^d)$ be a family of observables
satisfying the condition
\begin{equation}\label{cond_com}
\big\{\{\Phi_j,H\},H\big\}=0
\end{equation}
for each $j\in\{1,\ldots,d\}$. Then we have (see Theorem \ref{IntCont}, Corollary
\ref{cor_car} and Lemma \ref{LemChile} for a precise statement):

\begin{Theorem}\label{gros_bidule}
Let $H$ and $\Phi$ be as above. Let $f:\R^d\to\C$ be such that $f=1$ on a
neighbourhood of $\,0$, $f=0$ at infinity, and $f(x)=f(-x)$ for each $x\in\R^d$. Then
there exist a closed subset $\Crit(H,\Phi)\subset M$ and an observable 
$T_f\in\cinf\big(M\setminus\Crit(H,\Phi)\big)$ satisfying $\{T_f,H\}=1$ on
$M\setminus\Crit(H,\Phi)$ such that
\begin{equation}\label{LaFormule}
\lim_{r\to\infty}\12\int_0^\infty\d t\,\big[\big(f(\Phi/r)\circ\varphi_{-t}\big)(m)
-\big(f(\Phi/r)\circ\varphi_t\big)(m)\big]
=T_f(m)
\end{equation}
for each $m\in M\setminus\Crit(H,\Phi)$.
\end{Theorem}

The observable $T_f$ admits a very simple expression given in terms of the Poisson
brackets $\partial_jH:=\{\Phi_j,H\}$ and the vector
$\nabla H:=(\partial_1H,\ldots,\partial_dH)$, namely,
\begin{equation}\label{prem_def_T}
T_f:=-\Phi\cdot(\nabla R_f)(\nabla H),
\end{equation}
where $\nabla R_f:\R^d\to\C^d$ is some explicit function (see Section
\ref{section_R}).

In order to give an interpretation of Formula \eqref{LaFormule}, consider for a
moment the situation where $M:=T^*\R^n\simeq\R^{2n}$ is the standard symplectic
manifold with canonical coordinates $(q,p)$ and $2$-form
$\omega:=\sum_{j=1}^n\d q^j\wedge\d p_j$. Furthermore, let $H(q,p):=h(p)$ be a purely
kinetic energy Hamiltonian and let $\Phi(q,p):=q$ be the standard family of position
observables. In such a case, the condition \eqref{cond_com} is readily verified, the
vector $\nabla H$ reduces to the usual velocity observable $\nabla h$ associated to
$H$, and the l.h.s. of Formula \eqref{LaFormule} has the following meaning: For $r>0$
and $m\in M\setminus\Crit(H,\Phi)$ fixed, it is equal to the difference of times
spent by the classical orbit $\{\varphi_t(m)\}_{t\in\R}$ in the past (first term) and
in the future (second term) within the region $\Sigma_r:=\supp[f(\Phi/r)]\subset M$
defined by the localisation observable $f(\Phi/r)$. Moreover, if we interpret the map
$\frac\d{\d H}:=\{T_f,\;\!\cdot\;\!\}$ as a derivation on
$\cinf\big(M\setminus\Crit(H,\Phi)\big)$, then $T_f$ on the r.h.s. of
\eqref{LaFormule} can be seen as an observable ``derivative with respect to the
energy $H$'' on $M\setminus\Crit(H,\Phi)$, since $\frac\d{\d H}(H)=\{T_f,H\}=1$ on 
$M\setminus\Crit(H,\Phi)$. Therefore, Formula \eqref{LaFormule} provides a new
relation between sojourn times and variation of energy along classical orbits.
Evidently, this interpretation remains valid in the general case provided that we
consider the observables $\Phi_j$ as the components of an abstract position
observable $\Phi$ (see Remark \ref{Rem_Int}).

Our interest in this issue has been aroused by a recent paper \cite{RT10}, where the
authors establish a similar formula in the framework of quantum (Hilbertian) theory.
In that reference, $H$ is a selfadjoint operator in a Hilbert space $\H$, 
$\Phi\equiv(\Phi_1,\ldots,\Phi_d)$ is a family of mutually commuting selfadjoint
operators in $\H$, \eqref{cond_com} is a suitable version of the commutation relation
$\big[[\Phi_j,H],H\big]=0$, and $T_f$ is a time operator for $H$ (\ie a symmetric
operator satisfying the canonical commutation relation $[T_f,H]=i$). So, apart from
its intrinsic interest, the present paper provides also a new example of result valid
both in quantum and classical mechanics. Points of the symplectic manifold correspond
to vectors of the Hilbert space, complete Hamiltonian flows correspond to
one-parameter unitary groups, Poisson brackets correspond to commutators of
operators, \etc~(see \cite[Sec.~5.4]{AM78} and \cite{Lan98} for the interconnections
between classical and quantum mechanics). Accordingly, we try put into light
throughout all of the paper the relation between both theories. For instance, we link
in Remark \ref{rem_spec} the confinement (resp. the non-periodicity) of the
classical orbits $\{\varphi_t(m)\}_{t\in\R}$, $m\in M$, to the affiliation of the
corresponding quantum orbits $\{\e^{itH}\psi\}_{t\in\R}$, $\psi\in\H$, to the
singular (resp. absolutely continuous) subspace of $\H$. Moreover, we show in Section
\ref{qsys} that the Hilbertian space theory of \cite{RT10} can be recast into the
present framework of symplectic geometry by using expectation values.

We also mention that Formula \eqref{LaFormule}, with r.h.s. defined by
\eqref{prem_def_T}, provides a crucial preliminary step for the proof of the
existence of classical time delay for abstract scattering pairs $\{H,H+V\}$ (see
\cite{BP09}, \cite[Sec.~4.1]{dCN02}, and \cite[Sec.~3.4]{Thi97} for an account on
classical time delay). If $V$ is an appropriate perturbation of $H$ and $S$ is the
associated scattering map, then the classical time delay $\tau(m)$ for $m\in M$
defined in terms of the localisation operators $f(\Phi/r)$ should be reexpressed as
follows: it is equal to the l.h.s. of \eqref{LaFormule} minus the same quantity with
$m$ replaced by $S(m)$. Therefore, if $m$ and $S(m)$ are elements of
$M\setminus\Crit(H,\Phi)$, then the classical time delay for the scattering pair
$\{H,H+V\}$ should satisfy the equation
$$
\tau(m)=(T_f-T_f\circ S)(m).
$$
Now, the property $\{T_f,H\}(m)=1$ implies that $T_f(m)=(T_f\circ\varphi_t)(m)-t$ for
each $t\in\R$. Since $S$ commutes with $\varphi_t$, this would imply that
$$
\tau(m)=\big[(T_f-T_f\circ S)\circ\varphi_t\big](m)
$$
for all $t\in\R$, meaning that the classical time delay is a first integral of the
free motion. This property corresponds in the quantum case to the fact that the time
delay operator is decomposable in the spectral representation of the free Hamiltonian
(see \cite[Rem.~4.4]{RT11}).

Let us now describe more precisely the content of this paper. In Section
\ref{section_R} we recall some definitions in relation with the function $f$ that
appear in Theorem \ref{gros_bidule}. The function $R_f$ is introduced and some of its
properties are presented. Then we prove various versions of Formula \eqref{LaFormule}
in the particular case where the functions $\Phi\circ\varphi_{\pm t}:M\to\R^d$ are
fixed vectors $x\pm ty$, $x,y\in\R^d$ (see Proposition \ref{f_integrale}, Lemma
\ref{boulette} and Corollary \ref{poiscaille}).

In Section \ref{Sec_Crit}, we introduce the Hamiltonian system $(M,\omega,H)$ and the
abstract position observable $\Phi$. Then we define the (closed) set of critical
points $\Crit(H,\Phi)$ associated to $H$ and $\Phi$ as (see \cite[Def.~2.5]{RT10} for
the quantum analogue):
$$
\Crit(H,\Phi):=\big\{m\in M\mid(\nabla H)(m)=0\big\}.
$$
When $H(q,p)=h(p)$ and $\Phi(q,p):=q$ on $M=\R^{2n}$, $\Crit(H,\Phi)$ coincides with
the usual set $\Crit(H)$ of critical points of the Hamiltonian vector field $X_H$,
\ie
$$
\Crit(H)
\equiv\big\{m\in M\mid X_H(m)=0\big\}
=\big\{(q,p)\in\R^{2n}\mid(\nabla h)(p)=0\big\}
=\Crit(H,\Phi).
$$
But, in general, we simply have the inclusion $\Crit(H)\subset\Crit(H,\Phi)$.

In Section \ref{sec_main_form}, we prove the main results of this paper. Namely, we
show Formula \eqref{LaFormule} when the localisation function $f$ is regular (Theorem
\ref{IntCont}) or equal to a characteristic function (Corollary \ref{cor_car}). We
also establish in Theorem \ref{thm_discrete} a discrete-time version of Formula
\eqref{LaFormule}. The interpretation of these results is discussed in Remarks
\ref{rem_spec} and \ref{Rem_Int}.

In Section \ref{exemp}, we show that our results apply to many Hamiltonian systems
$(M,\omega,H)$ appearing in literature. In the case of finite-dimensional manifolds,
we treat, among other examples, Stark Hamiltonians, homogeneous Hamiltonians, purely
kinetic Hamiltonians, the repulsive harmonic potential, the simple pendulum, central
force systems, the Poincar\'e ball model and covering manifolds. In the case of
infinite-dimensional manifolds, we discuss separately classical and quantum
Hamiltonians systems. In the classical case, we treat the wave equation, the
nonlinear Schr\"odinger equation and the Korteweg-de Vries equation. In the quantum
case, we explain how to recast into our framework the (Hilbertian) examples of
\cite[Sec.~7]{RT10}, and we also treat an example of Laplacian on trees and complete
Fock spaces. In all these cases, we are able to exhibit a family of position
observables $\Phi$ satisfying our assumptions. The diversity of the examples covered
by our theory, together with the existence of a quantum analogue \cite{RT10}, make us
strongly believe that Formula \eqref{LaFormule} is of natural character. Moreover it
also suggests that the existence of time delay is a very common feature of classical
scattering theory.

%--------------------------------------------------------------------------------------
\section{Integral formula}\label{section_R}
\setcounter{equation}{0}
%--------------------------------------------------------------------------------------

In this section, we prove an integral formula and a summation formula for functions
on $\R^d$. For this, we start by recalling some properties of a class of averaged
localisation functions which appears naturally when dealing with quantum scattering
theory. These functions, which are denoted $R_f$, are constructed in terms of
functions $f\in\linf(\R^d)$ of localisation around the origin $0\in\R^d$. They were
already used, in one form or another, in \cite{GT07,RT10,RT11,Tie08,Tie09}. We use
the notation $\langle x\rangle:=\sqrt{1+|x|^2}$ for any $x\in\R^d$.

\begin{Assumption}\label{assumption_f}
The function $f\in\linf(\R^d)$ satisfies the following conditions:
\begin{enumerate}
\item[(i)] There exists $\rho>0$ such that
$|f(x)|\le{\rm Const.}\;\!\langle x\rangle^{-\rho}$ for almost every $x\in\R^d$.
\item[(ii)] $f=1$ on a neighbourhood of~~$0$.
\end{enumerate}
\end{Assumption}

It is clear that $\lim_{r\to\infty}f(x/r)=1$ for each $x\in\R^d$ if $f$ satisfies Assumption \ref{assumption_f}. Furthermore, one has for each $x\in \R^d\setminus\{0\}$
$$
\left|\int_0^\infty\frac{\d\mu}\mu\[f(\mu x)-\chi_{[0,1]}(\mu)\]\right|
\le\int_0^1\frac{\d\mu}\mu\,|f(\mu x)-1|
+{\rm Const.}\int_1^{+\infty}\d\mu\,\mu^{-(1+\rho)}
<\infty,
$$
where $\chi_{[0,1]}$ denotes the characteristic function for the interval $[0,1]$.
Therefore the function $R_f:\R^d\setminus\{0\}\to\C$ given by
$$
R_f(x):=\int_0^{+\infty}\frac{\d\mu}\mu\[f(\mu x)-\chi_{[0,1]}(\mu)\]
$$
is well-defined.

In the next lemma we recall some differentiability and homogeneity properties of
$R_f$. We also give the explicit form of $R_f$ when $f$ is a radial function. The
reader is referred to \cite[Sec. 2]{Tie09} for proofs and details. The symbol
$\S(\R^d)$ stands for the Schwartz space on $\R^d$.

\begin{Lemma}\label{function_R}
Let $f$ satisfy Assumption \ref{assumption_f}.
\begin{enumerate}
\item[(a)] Assume that $\frac{\partial f}{\partial x_j}(x)$ exists for all
$j\in\{1,\ldots,d\}$ and $x\in\R^d$, and suppose that there exists some $\rho>0$
such that
$
\big|\frac{\partial f}{\partial x_j}(x)\big|
\le{\rm Const.}\;\!\langle x\rangle^{-(1+\rho)}
$
for each $x\in\R^d$. Then $R_f$ is differentiable on $\R^d\setminus\{0\}$, and
its gradient is given by
\begin{equation*}
(\nabla R_f)(x)=\int_0^\infty\d\mu\,(\nabla f)(\mu x).
\end{equation*}
In particular, if $f\in\S(\R^d)$ then $R_f$ belongs to $\cinf(\R^d\setminus\{0\})$.
\item[(b)] Assume that $R_f$ belongs to ${\sf C}^m(\R^d\setminus\{0\})$ for some
$m\ge1$. Then one has for each $x\in \R^d\setminus\{0\}$ and $t>0$ the homogeneity
properties
\begin{align}
x\cdot(\nabla R_f)(x)&=-1,\label{minusone}\\
t^{|\alpha|}(\partial^\alpha R_f)(tx)
&=(\partial^\alpha R_f)(x),\nonumber
\end{align}
where $\alpha\in\N^d$ is a multi-index with $1\le|\alpha|\le m$.
\item[(c)] Assume that $f$ is radial, \ie there exists $f_0\in\linf(\R)$ such that
$f(x)=f_0(|x|)$ for almost every $x\in \R^d$. Then $R_f$ belongs to
$\cinf(\R^d\setminus\{0\})$, and $(\nabla R_f)(x)=-x^{-2}x$.
\end{enumerate}
\end{Lemma}

In the sequel, we say that a function $f:\R^d\to\C$ is even if $f(x)=f(-x)$ for
almost every $x\in\R^d$.

\begin{Proposition}\label{f_integrale}
Let $f:\R^d\to\C$ be an even function as in Lemma \ref{function_R}.(a). Then we
have for each $x\in\R^d$ and each $y\in\R^d\setminus\{0\}$
\begin{equation}\label{tite_formule}
\lim_{r\to\infty}\12\int_0^\infty\d t\;\!
\Big[f\Big(\frac{x-ty}r\Big)-f\Big(\frac{x+ty}r\Big)\Big]
=-x\cdot(\nabla R_f)(y).
\end{equation}
In particular, if $f$ is radial, the l.h.s. is independent of $f$ and equal to
$(x\cdot y)/y^2$.
\end{Proposition}

\begin{proof}
The change of variables $\mu:=t/r$, $\nu:=1/r$, and the fact that $f$ is even,
gives
\begin{align}
&\lim_{r\to\infty}\12\int_0^\infty\d t\,\textstyle
\big[f\big(\frac{x-ty}r\big)-f\big(\frac{x+ty}r\big)\big]\nonumber\\
&=\lim_{\nu\searrow0}\12\int_0^\infty\frac{\d\mu}\nu\,\big[
f(\nu x-\mu y)-f(\nu x+\mu y)\big]\nonumber\\
&=\lim_{\nu\searrow0}\12\int_0^\infty\d\mu\,\textstyle
\big\{\frac1\nu\big[f(\nu x-\mu y)-f(-\mu y)\big]-\frac1\nu\big[f(\nu x+\mu y)
-f(\mu y)\big]\big\}.\label{integrant}
\end{align}
By using the mean value theorem and the assumptions of Lemma \ref{function_R}.(a),
one obtains that
$$
{\textstyle\frac1\nu}\big|f(\nu x\pm\mu y)-f(\pm\mu y)\big|
\le{\rm Const.}\sup_{\xi\in[0,1]}\big\langle\xi\nu x\pm\mu y\big\rangle^{-(1+\rho)}
$$
for some $\rho>0$. Therefore, if $\mu$ is big enough, the integrant in
\eqref{integrant} is bounded by
$$
{\rm Const.}\;\!\big\langle\mu|y|-|x|\big\rangle^{-(1+\rho)}.
$$
for all $\nu\in(0,1)$. This implies that the integrant in \eqref{integrant} is
bounded uniformly in $\nu\in(0,1)$ by a function belonging to
$\lone\big([0,\infty),\d\mu\big)$. So, we can apply Lebesgue's dominated convergence
theorem to interchange the limit on $\nu$ with the integration over $\mu$ in
\eqref{integrant}. This, together with the fact that $(\nabla f)(-x)=-(\nabla f)(x)$,
leads to the desired result:
\begin{align*}
\lim_{r\to\infty}\12\int_0^\infty\d t\,\textstyle
\big[f\big(\frac{x-ty}r\big)-f\big(\frac{x+ty}r\big)\big]
&=\displaystyle\12\int_0^\infty\d\mu\,
\big[x\cdot(\nabla f)(-\mu y)-x\cdot(\nabla f)(\mu y)\big]\\
&=-\int_0^\infty\d\mu\,x\cdot(\nabla f)(\mu y)\\
&=-x\cdot(\nabla R_f)(y).\qedhere
\end{align*}
\end{proof}
The result of Proposition \ref{f_integrale} can be extended to less regular functions
$f:\R^d\to\C$. The interested reader can check that the result holds for functions
$f$ admitting a weak derivative $f'$ such that, for every real line $L\subset\R^d$,
$f'$ is of class $\lone$ on $L$ (see \cite[Thm.~2.1.6]{Zie89}). We only present here
the case (of particular interest for the theory of classical time delay) where $f$ is
the characteristic function $\chi_1$ for the unit ball
$B_1:=\{x\in\R^d\mid|x|\le 1\}$.

\begin{Lemma}\label{boulette}
One has for each $x\in\R^d$ and each $y\in\R^d\setminus\{0\}$
$$
\lim_{r\to\infty}\12\int_0^\infty\d t\;\!
\Big[\chi_1\Big(\frac{x-ty}r\Big)-\chi_1\Big(\frac{x+ty}r\Big)\Big]
=\frac{x\cdot y}{y^2}\;\!.
$$
\end{Lemma}

\begin{proof}
Direct calculations and the change of variables $\mu:=t/r$, $\nu:=1/r$, give
\begin{align*}
\int_0^\infty\d t\,\textstyle\chi_1\big(\frac{x\pm ty}r\big)
=\int_0^\infty\frac{\d\mu}\nu~\chi_{[0,1]}\big(|\nu x\pm\mu y|^2\big)
&=\int_0^\infty\d\mu~\textstyle\chi_{[0,y^{-2}]}
\Big(\frac{\nu^2x^2}{y^2}\pm\frac{2\nu\mu x\cdot y}{y^2}+\mu^2\Big)\\
&=\int_0^\infty\frac{\d\mu}\nu~\textstyle\chi_{[0,y^{-2}]}
\Big(\big(\mu\pm\frac{\nu x\cdot y}{y^2}\big)^2
+\frac{\nu^2}{y^4}\big(x^2y^2-(x\cdot y)^2\big)\Big)\\
&=\int_0^\infty\frac{\d\mu}\nu~\textstyle\chi_{[-a(\nu,x,y),y^{-2}-a(\nu,x,y)]}
\Big(\big(\mu\pm\frac{\nu x\cdot y}{y^2}\big)^2\Big),
\end{align*}
with $a(\nu,x,y):=\frac{\nu^2}{y^4}\big(x^2y^2-(x\cdot y)^2\big)$. Now,
$a(\nu,x,y)\ge0$, and $y^{-2}-a(\nu,x,y)\ge0$ if $\nu>0$ is small enough. So,
the last expression is equal to
\begin{align*}
\int_0^\infty\frac{\d\mu}\nu~\textstyle\chi_{[0,y^{-2}-a(\nu,x,y)]}
\Big(\big(\mu\pm\frac{\nu x\cdot y}{y^2}\big)^2\Big)
&=\int_0^\infty\frac{\d\mu}\nu~\textstyle
\chi_{\big[-\sqrt{y^{-2}-a(\nu,x,y)}\mp\frac{\nu x\cdot y}{y^2},
\sqrt{y^{-2}-a(\nu,x,y)}\mp\frac{\nu x\cdot y}{y^2}\big]}(\mu)\\
&=\frac1\nu\sqrt{y^{-2}-a(\nu,x,y)}\mp\frac{x\cdot y}{y^2}
\end{align*}
if $\nu$ is small enough. This implies that
\begin{align*}
&\lim_{r\to\infty}\12\int_0^\infty\d t\,\textstyle
\big[\chi_1\big(\frac{x-ty}r\big)-\chi_1\big(\frac{x+ty}r\big)\big]\\
&=\displaystyle\lim_{\nu\searrow0}\frac12
\Big(\frac1\nu\sqrt{y^{-2}-a(\nu,x,y)}+\frac{x\cdot y}{y^2}
-\frac1\nu\sqrt{y^{-2}-a(\nu,x,y)}+\frac{x\cdot y}{y^2}\Big)\\
&=\frac{x\cdot y}{y^2}\;\!.\qedhere
\end{align*}
\end{proof}

For the next corollary, we need the following version of the Poisson summation
formula (see \cite[Thm.~5]{DF37} or \cite[Thm.~45]{Tit48}).

\begin{Lemma}\label{poisson}
Let $g:(0,\infty)\to\C$ be a continuous function of bounded variation in
$(0,\infty)$. Suppose that $\lim_{t\to\infty}g(t)=0$ and that the improper Riemann
integral $\int_0^\infty\d t\,g(t)$ exists. Then we have the identity
$$
\12\;\!g(0)+\sum_{n\ge1}g(n)
=\int_0^\infty\d t\,g(t)+2\sum_{n\ge1}\int_0^\infty\d t\,\cos(2\pi nt)g(t).
$$
\end{Lemma}

\begin{Corollary}\label{poiscaille}
Let $f:\R^d\to\C$ be an even function such that
\begin{enumerate}
\item[(i)] $f=1$ on a neighbourhood of $\,0$.
\item[(ii)] For each $\alpha\in\N^d$ with $|\alpha|\le2$, the derivative
$\partial^\alpha f$ exists and satisfies
$
|(\partial^\alpha f)(x)|\le{\rm Const.}\;\!\langle x\rangle^{-(1+\rho)}
$
for some $\rho>0$ and all $x\in\R^d$.
\end{enumerate}
Then we have for each $x\in\R^d$ and each $y\in\R^d\setminus\{0\}$
\begin{equation}\label{formulette}
\lim_{r\to\infty}\12\sum_{n\ge1}
\Big[f\Big(\frac{x-ny}r\Big)-f\Big(\frac{x+ny}r\Big)\Big]
=-x\cdot(\nabla R_f)(y).
\end{equation}
In particular, if $f$ is radial, the l.h.s. is independent of $f$ and equal to
$(x\cdot y)/y^2$.	
\end{Corollary}

\begin{proof}
For $r>0$ given, the function
$$
\textstyle g_r:(0,\infty)\to\C,\quad t\mapsto
g_r(t):=f\big(\frac{x-ty}r\big)-f\big(\frac{x+ty}r\big),
$$
satisfies all the hypotheses of Lemma \ref{poisson}. Thus
$$
\lim_{r\to\infty}\12\sum_{n\ge1}\textstyle
\big[f\big(\frac{x-ny}r\big)-f\big(\frac{x+ny}r\big)\big]
=\displaystyle\lim_{r\to\infty}\12\int_0^\infty\d t\,g_r(t)
+\lim_{r\to\infty}\sum_{n\ge1}\int_0^\infty\d t\,\cos(2\pi nt)g_r(t).
$$
The first term is equal to $-x\cdot(\nabla R_f)(y)$ due to Proposition
\ref{f_integrale}. For the second term, the change of variables $\mu:=t/r$,
$\nu:=1/r$, and two integrations by parts give
\begin{align*}
&\lim_{r\to\infty}\sum_{n\ge1}\int_0^\infty\d t\,\cos(2\pi nt)g_r(t)\\
&=\lim_{\nu\searrow0}\sum_{n\ge1}\int_0^\infty\frac{\d\mu}\nu\,
\cos(2\pi n\mu/\nu)
\big[f(\nu x-\mu y)-f(\nu x+\mu y)\big]\\
&=\sum_jy_j\lim_{\nu\searrow0}\sum_{n\ge1}\int_0^\infty\d\mu\,
\frac{\sin(2\pi n\mu/\nu)}{2\pi n}
\bigg(\frac{\partial f}{\partial x_j}(\nu x-\mu y)
+\frac{\partial f}{\partial x_j}(\nu x+\mu y)\bigg)\\
&=\sum_jy_j\lim_{\nu\searrow0}\sum_{n\ge1}\frac{2\nu}{(2\pi n)^2}
\frac{\partial f}{\partial x_j}\big(\nu x\big)\\
&\quad-\sum_{j,k}y_jy_k\lim_{\nu\searrow0}\sum_{n\ge1}\int_0^\infty\d\mu\,
\frac{\nu\cos(2\pi n\mu/\nu)}{(2\pi n)^2}
\bigg(\frac{\partial^2f}{\partial x_k\partial x_j}(\nu x-\mu y)
+\frac{\partial^2f}{\partial x_k\partial x_j}(\nu x+\mu y)\bigg).
\end{align*}
Since $\sum_{n\ge1}1/n^2<\infty$, one sees directly that the first term is equal to
zero. Using the fact that
$
\big|\frac{\partial^2f}{\partial x_k\partial x_j}(x)\big|
\le{\rm Const.}\<x\>^{-(1+\rho)}
$
for some $\rho>0$ and all $x\in\R^d$, one also obtains that the second term is equal
to zero. Therefore,
$$
\lim_{r\to\infty}\12\sum_{n\ge1}\textstyle
\big[f\big(\frac{x-ny}r\big)-f\big(\frac{x+ny}r\big)\big]
=-x\cdot(\nabla R_f)(y),
$$
and the claim is proved.
\end{proof}

%--------------------------------------------------------------------------------------
\section{Hamiltonian dynamics}\label{sec_ham}
\setcounter{equation}{0}
%--------------------------------------------------------------------------------------

In the sequel, we require the presence of a symplectic structure in order to speak of
Hamiltonian dynamics. However our results still hold if one is only given a Poisson
structure. A lack of examples and some complications in infinite dimension regarding
the identification of vector fields with derivations have led us to restrict the
discussion to the symplectic case for the sake of clarity.

%--------------------------------------------------------------------------------------
\subsection{Critical points}\label{Sec_Crit}
%--------------------------------------------------------------------------------------

Let $M$ be a symplectic manifold, \ie a smooth  manifold endowed with a closed
two-form $\omega$ such that the morphism
$
TM\ni X\mapsto\omega^\flat(X):=\iota_X\omega
$
is an isomorphism. In infinite dimension, such a manifold is said to be a strong
symplectic manifold (in opposition to a weak symplectic manifold, when the above map
is only injective; see \cite[Sec.~8.1]{AMR88}). When the dimension is finite, the
dimension must be even, say equal to $2n$, and the $2n$-form $\omega^n$ must be a
volume form. The Poisson bracket is defined as follows: for each $f\in\cinf(M)$ we
define the vector field $X_f:=(\omega^\flat)^{-1}(\d f)$, \ie
$\d f(\;\!\cdot\;\!)=\omega(X_f,\;\!\cdot\;\!)$, and set $\{f,g\}:=\omega(X_f,X_g)$
for each $f,g\in\cinf(M)$.

In the sequel, the function $H\in\cinf(M)$ is an Hamiltonian with complete vector
field $X_H$. So, the flow $\{\varphi_t\}$ associated to $H$ is defined for all
$t\in\R$, it preserves the Poisson bracket:
$$
\big\{f\circ\varphi_t,g\circ\varphi_t\big\}=\{f,g\}\circ\varphi_t,\quad t\in\R,
$$
and satisfies the usual evolution equation:
\begin{equation}\label{evolution}
\frac\d{\d t}\;\!f\circ\varphi_t=\{f,H\}\circ\varphi_t,\quad t\in\R.
\end{equation}
In particular, the Hamiltonian $H$ is preserved along its flow, \ie
$H\circ\varphi_t=H$ for all $t\in\R$. We also consider an abstract family 
$\Phi\equiv(\Phi_1,\ldots,\Phi_d)\in\cinf(M;\R^d)$ of observables\footnote{If need
be, the results of this article can be extended to the case where $H$ and $\Phi_j$
are functions of class $\cone$ with $\{\Phi_j,H\}$ also $\cone$.}, and define the
associated functions
$$
\partial_jH:=\{\Phi_j,H\}\in\cinf(M)\qquad\hbox{and}\qquad
\nabla H:=(\partial_1H,\ldots,\partial_dH)\in\cinf(M;\R^d).
$$
Then, one can introduce a natural set of critical points:

\begin{Definition}[Critical points]
The set
$$
\Crit(H,\Phi):=(\nabla H)^{-1}(\{0\})\subset M
$$
is called the set of critical points associated to $H$ and $\Phi$.
\end{Definition}

The set $\Crit(H,\Phi)$ is closed in $M$ since $\nabla H$ is continuous. Furthermore,
since $\{\Phi_j,H\}=\d\Phi_j(X_H)$, the set
$$
\Crit(H):=\big\{m\in M\mid X_H(m)=0\big\}\equiv\big\{m\in M\mid \d H_m=0\big\}
$$
of usual critical points of $H$ satisfies the inclusion
$\Crit(H)\subset\Crit(H,\Phi)$.

Our main assumption is the following:

\begin{Assumption}\label{AssCom}
One has $\big\{\{\Phi_j,H\},H\big\}=0\,$ for each $j\in\{1,\ldots,d\}$.
\end{Assumption}

Assumption \ref{AssCom} imposes that all the brackets $\{\Phi_j,H\}$ are first
integrals of the motion given by $H$. When $M$ is a symplectic manifold of dimension
$2n$, these first integrals are functions of $k \in \{1,2,\ldots,2n-1\}$ independent
first integrals $J_1\equiv H,J_2,\ldots,J_k$ ($J_1,\ldots,J_k$ are independent in the
sense that their differential are linearly independent at each point of
$M$)\footnote{In the setup of Liouville's theorem \cite[Sec.~49]{Arn89}, we have
$k=n$ and the first integrals are mutually in involution. Furthermore, on the
connected components of submanifolds given by fixing the values of these $n$
integrals in involution, the flow is conjugate to a translation flow on cylinders
$\R^{n-\ell}\times\mathbb T^\ell$ (see \cite[Thm.~5.2.24]{AM78}).}. So, one should
have $\{\Phi_j,H\}=g_j(J_1,\ldots,J_k)$ for some functions $g_j\in\cinf(\R^n;\R)$.
Using the properties of $\{\;\!\cdot\;\!,H\}$ as a derivation, one infers that
$\big\{g_j(J_1,\ldots,J_k)^{-1}\Phi_j,H\big\}=1$ outside
$g_j(J_1,\ldots,J_k)^{-1}(\{0\})$. Thus, if $k$ first integrals as $J_1,\ldots,J_k$
are known, finding functions $\Phi_j$ satisfying Assumption \ref{AssCom} is to some
extent equivalent to finding functions $\Phi_0$ solving $\{\Phi_0,H\}=1$ (the
equivalence is not complete because these functions $\Phi_0$ are in general not
$\cinf$ since $\{\;\!\cdot\;\!,H\}$ is necessarily $0$ on $\Crit(H)$).

For further use, we define the $\cinf$-function $T_f:M\setminus\Crit(H,\Phi)\to\R$ by
$$
T_f:=-\Phi\cdot(\nabla R_f)(\nabla H).
$$
When $f$ is radial, $T_f$ is independent of $f$ and equal to
$$
T:=\Phi\cdot\frac{\nabla H}{(\nabla H)^2}\;\!,
$$
due to Lemma \ref{function_R}.(c).
In fact, the closed subset $T^{-1}(\{0\})$ of $M\setminus\Crit(H,\Phi)$ admits an
interesting interpretation: If we consider the observables $\Phi_j$ as the components
of an abstract position observable $\Phi$, then $\nabla H$ can be seen as the
velocity vector for the Hamiltonian $H$, and the condition
\begin{equation}\label{truffette}
T(m)=0~\iff~\Phi(m)\cdot(\nabla H)(m)=0
\end{equation}
means that the position and velocity vectors are orthogonal at $m\in T^{-1}(\{0\})$.
Alternatively, one has $T(m)=0$ if and only if the vector fields $X_{|\Phi|^2}$ and
$X_H$ are $\omega$-orthogonal at $m$, that is,
$\omega_m\big(X_{|\Phi|^2}(m),X_H(m)\big)=0$. The simplest example illustrating the
condition \eqref{truffette} is when $\Phi(q,p):=q$ and $H(q,p):=\12|p|^2$ are the
usual position and kinetic energy on
$(M,\omega):=\big(\R^{2n},\sum_{j=1}^n\d q^j\wedge\d p_j\big)$. In such a case,
\eqref{truffette} reduces to $q\cdot p=0$.

%--------------------------------------------------------------------------------------
\subsection{Sojourn times of classical orbits}\label{sec_main_form}
%--------------------------------------------------------------------------------------

Next Theorem is our main result. We refer to Remark \ref{Rem_Int} below for its
interpretation.

\begin{Theorem}\label{IntCont}
Let $H$ and $\Phi$ satisfy Assumption \ref{AssCom}. Let $f:\R^d\to\C$ be an even
function as in Lemma \ref{function_R}.(a). Then we have for each point
$m\in M\setminus\Crit(H,\Phi)$
\begin{equation}\label{BelleLouloute}
\lim_{r\to\infty}\12\int_0^\infty\d t\,\big[\big(f(\Phi/r)\circ\varphi_{-t}\big)(m)
-\big(f(\Phi/r)\circ\varphi_t\big)(m)\big]
=T_f(m).
\end{equation}
In particular, if $f$ is radial, the l.h.s. is independent of $f$ and equal to
$\Phi(m)\cdot\frac{(\nabla H)(m)}{(\nabla H)(m)^2}$\;\!.
\end{Theorem}

\begin{proof}
Equation \eqref{evolution} implies that
$$
\frac\d{\d t}\;\!\Phi_j\circ\varphi_t=\{\Phi_j,H\}\circ\varphi_t
$$
for each $t\in\R$. Similarly, using Assumption \ref{AssCom}, one gets that
$$
\frac\d{\d t}\;\!\{\Phi_j,H\}\circ\varphi_t
=\big\{\{\Phi_j,H\},H\big\}\circ\varphi_t
=0.
$$
So, $\Phi_j$ varies linearly in $t$ along the flow of $X_H$, and one gets for any
$m\in M$
$$
(\Phi_j\circ\varphi_t)(m)
=(\Phi_j\circ\varphi_0)(m)
+t\,\Big(\frac\d{\d t}\;\!(\Phi_j\circ\varphi_t)(m)\Big|_{t=0}\Big)
=\Phi_j(m)+t(\partial_jH)(m).
$$
This, together with Formula \eqref{tite_formule}, gives
\begin{align*}
&\lim_{r\to\infty}\12\int_0^\infty\d t\,
\big[\big(f(\Phi/r)\circ\varphi_{-t}\big)(m)
-\big(f(\Phi/r)\circ\varphi_t\big)(m)\big]\\
&=\lim_{r\to\infty}\12\int_0^\infty\d t\,\textstyle
\Big[f\Big(\frac{\Phi(m)-t(\nabla H)(m)}r\Big)
-f\Big(\frac{\Phi(m)+t(\nabla H)(m)}r\Big)\Big]\\
&=T_f(m).\qedhere
\end{align*}
\end{proof}

Due to Lemma \ref{boulette}, the proof of Theorem \ref{IntCont} also works in the
case $f=\chi_1$. So, we have the following corollary.

\begin{Corollary}\label{cor_car}
Let $H$ and $\Phi$ satisfy Assumption \ref{AssCom}. Then we have for each point
$m\in M\setminus\Crit(H,\Phi)$
\begin{equation}\label{2eme_loulette}
\lim_{r\to\infty}\12\int_0^\infty\d t\,\textstyle
\big[\big(\chi_1(\Phi/r)\circ\varphi_{-t}\big)(m)
-\big(\chi_1(\Phi/r)\circ\varphi_t\big)(m)\big]
=\Phi(m)\cdot\frac{(\nabla H)(m)}{(\nabla H)(m)^2}\;\!.
\end{equation}
\end{Corollary}

We know from the proof of Theorem \ref{IntCont} that
\begin{equation}\label{linear}
(\Phi_j\circ\varphi_t)(m)=\Phi_j(m)+t(\partial_jH)(m)
\quad\hbox{for all $t\in\R$ and all $m\in M$.}
\end{equation}
Therefore, the l.h.s. of \eqref{BelleLouloute} and \eqref{2eme_loulette} are zero if
$m\in\Crit(H,\Phi)$.

For the next remark, we recall that any selfadjoint operator $A$ in a Hilbert space
$\H$, with spectral measure $E^A(\;\!\cdot\;\!)$, is reduced by an orthogonal
decomposition \cite[Sec. 7.4]{Wei80}
$$
\H
=\H_{\rm ac}(A)\oplus\H_{\rm p}(A)\oplus\H_{\rm sc}(A)
\equiv\H_{\rm ac}(A)\oplus\H_{\rm s}(A),
$$
where $\H_{\rm ac}(A),\H_{\rm p}(A),\H_{\rm sc}(A)$ and $\H_{\rm s}(A)$ are
respectively the absolutely continuous, the pure point, the singular continuous and
the singular subspaces of $A$. Furthermore, a vector $\varphi\in\H$ is said to have
spectral support with respect to $A$ in a set $J\subset\R$ if
$\varphi=E^A(J)\varphi$.

\begin{Remark}\label{rem_spec}
If $m\in\Crit(H,\Phi)$, then one must have $\varphi_t(m)\in\Crit(H,\Phi)$ for all
$t\in\R$, since \eqref{linear} implies $(\partial_jH)(\varphi_t(m))=(\partial_jH)(m)$
for all $t\in\R$. Conversely, if $m\in M\setminus\Crit(H,\Phi)$, then one must have
$\varphi_t(m)\neq m$ for all $t\neq0$, since $\Phi$ cannot take two different values
at a same point. So, under Assumption \ref{AssCom}, each orbit
$\{\varphi_t(m)\}_{t\in\R}$ either stays in $\Crit(H,\Phi)$ if $m\in\Crit(H,\Phi)$,
or stays outside $\Crit(H,\Phi)$ and is not periodic if $m\notin\Crit(H,\Phi)$.

In the corresponding Hilbertian framework \cite{RT10}, the Hamiltonian $H$ and the
functions $\Phi_j$ are selfadjoint operators in a Hilbert space $\H$, and the
critical set $\kappa$ associated to $H$ and $\Phi$ is a closed subset of the spectrum
of $H$. Outside $\kappa$, the spectrum of $H$ is purely absolutely continuous
\cite[Thm.~3.6.(a)]{RT10}. Therefore, vectors $\psi\in\H$ having spectral support
with respect to $H$ in $\kappa$ belong to the singular subspace $\H_{\rm s}(H)$ of
$H$, and thus lead to orbits $\{\e^{itH}\psi\}_{t\in\R}$ confined in $\H_{\rm s}(H)$
(for instance, $\e^{itH}\psi$ stays in a one-dimensional subspace of $\H$ if $\psi$
is an eigenvector of $H$). Conversely, vectors $\psi\in\H$ having spectral support
outside $\kappa$ belong to the absolute continuous subspace $\H_{\rm ac}(H)$ of $H$,
and thus lead to orbits $\{\e^{itH}\psi\}_{t\in\R}$ contained in $\H_{\rm ac}(H)$
(see \cite[Prop.~5.7]{Amr09} for the escape properties of such orbits). These
properties are the quantum counterparts of the confinement to $\Crit(H,\Phi)$
(when $m\in\Crit(H,\Phi)$) and the non-periodicity outside $\Crit(H,\Phi)$ (when
$m\notin\Crit(H,\Phi)$) of the classical orbits $\{\varphi_t(m)\}_{t\in\R}$.
\end{Remark}

\begin{Lemma}\label{LemChile}
If $H$, $\Phi$ and $f$ satisfy the assumptions of Theorem \ref{IntCont}, then we
have
\begin{equation}\label{E_derivative}
\{T_f,H\}\circ\varphi_t\equiv\frac\d{\d t}\;\!(T_f\circ\varphi_t)=1
\end{equation}
on $M\setminus\Crit(H,\Phi)$. In particular, one has $T_f\circ\varphi_t=T_f+t$ on 
$M\setminus\Crit(H,\Phi)$.
\end{Lemma}

If we interpret the map $\frac\d{\d H}:=\{T_f,\;\!\cdot\;\!\}$ as a derivation on
$\cinf\big(M\setminus\Crit(H,\Phi)\big)$, this implies that $T_f$ can be seen as an
observable ``derivative with respect to the energy $H$'' on $M\setminus\Crit(H,\Phi)$,
since
$$
\textstyle\frac\d{\d H}(H)=\{T_f,H\}=1
$$
on each orbit $\{\varphi_t(m)\}_{t\in\R}$, with $m\in M\setminus\Crit(H,\Phi)$.

\begin{proof}[Proof of Lemma \ref{LemChile}]
The first equality in \eqref{E_derivative} follows from \eqref{evolution}. For the
second one, we use successively the fact that $\varphi_t$ leaves invariant $H$ and
the Poisson bracket, Assumption \ref{AssCom}, and Equation \eqref{minusone}. Doing
so, we get on $M\setminus\Crit(H,\Phi)$ the following equalities
\begin{align*}
\frac\d{\d t}\;\!(T_f\circ\varphi_t)
=-\frac\d{\d t}\;\!(\Phi\circ\varphi_t)\cdot
(\nabla R_f)\big(\{\Phi\circ\varphi_t,H\}\big)
&=-\frac\d{\d t}\;\!\big(\Phi+t(\nabla H)\big)\cdot
(\nabla R_f)\big(\{\Phi+t(\nabla H),H\}\big)\\
&=-\frac\d{\d t}\;\!\big(\Phi+t(\nabla H)\big)\cdot(\nabla R_f)(\nabla H)\\
&=-(\nabla H)\cdot(\nabla R_f)(\nabla H)\\
&=1.\qedhere
\end{align*}
\end{proof}

\begin{Remark}\label{Rem_Int}
Theorem \ref{IntCont} relates the sojourn times of classical orbits within expanding
regions of $M$ to the observable $T_f$. If we consider the observables $\Phi_j$ as
the components of an abstract position observable $\Phi$, then the l.h.s. of Formula
\eqref{BelleLouloute} has the following meaning: For $r>0$ and
$m\in M\setminus\Crit(H,\Phi)$ fixed, it can be interpreted as the difference of
times spent by the classical orbit $\{\varphi_t(m)\}_{t\in\R}$ in the past (first
term) and in the future (second term) within the region
$\Sigma_r:=\supp[f(\Phi/r)]\subset M$ defined by the localisation observable
$f(\Phi/r)$. Thus, Formula \eqref{BelleLouloute} shows that this difference of times
tends as $r\to\infty$ to the value of the observable $T_f$ at $m$. Since $T_f$ can be
interpreted as an observable derivative with respect to the energy $H$, Formula
\eqref{BelleLouloute} provides a new relation between sojourn times and variation of
energy along classical orbits.
\end{Remark}

As a final result, we give a discrete-time counterpart of Theorem \ref{IntCont},
which could be of some interest in the context of approximation of symplectomorphisms
by time-$1$ maps of Hamiltonians flows (see \eg \cite{BG94},
\cite[Appendix~B]{Har00}, \cite{KP94} and references therein).

\begin{Theorem}\label{thm_discrete}
Let $H$ and $\Phi$ satisfy Assumption \ref{AssCom}. Let $f:\R^d\to\C$ be an even
function such that
\begin{enumerate}
\item[(i)] $f=1$ on a neighbourhood of $\,0$.
\item[(ii)] For each $\alpha\in\N^d$ with $|\alpha|\le2$, the derivative
$\partial^\alpha f$ exists and satisfies
$|(\partial^\alpha f)(x)|\le{\rm Const.}\<x\>^{-(1+\rho)}$ for some $\rho>0$ and all
$x\in\R^d$.
\end{enumerate}
Then we have for each point $m\in M\setminus\Crit(H,\Phi)$
$$
\lim_{r\to\infty}\12\sum_{n\ge1}\big[\big(f(\Phi/r)\circ\varphi_{-n}\big)(m)
-\big(f(\Phi/r)\circ\varphi_n\big)(m)\big]
=T_f(m).
$$
In particular, if $f$ is radial, the l.h.s. is independent of $f$ and equal to
$\Phi(m)\cdot\frac{(\nabla H)(m)}{(\nabla H)(m)^2}$\;\!.
\end{Theorem}

\begin{proof}
Let $m\in M\setminus\Crit(H,\Phi)$. Then we have by Equation \eqref{linear}
\begin{align*}
&\lim_{r\to\infty}\12\sum_{n\ge1}
\big[\big(f(\Phi/r)\circ\varphi_{-n}\big)(m)
-\big(f(\Phi/r)\circ\varphi_n\big)(m)\big]\\
&=\lim_{\nu\searrow0}\12\sum_{n\ge1}\textstyle
\Big[f\Big(\frac{\Phi(m)-n(\nabla H)(m)}r\Big)
-f\Big(\frac{\Phi(m)+n(\nabla H)(m)}r\Big)\Big],
\end{align*}
and the claim follows by Formula \eqref{formulette}.
\end{proof}

%--------------------------------------------------------------------------------------
\section{Examples}\label{exemp}
\setcounter{equation}{0}
%--------------------------------------------------------------------------------------

In this section we show that Assumption \ref{AssCom} is satisfied in various
situations. In these situations all the results of Section \ref{sec_ham} such as
Theorem \ref{IntCont} or Formula \eqref{E_derivative} hold. Some of the examples
presented here are the classical counterparts of examples discussed in
\cite[Sec.~7]{RT10} in the context of Hilbertian theory.

The configuration space of the system under consideration will sometimes be $\R^n$,
and the corresponding symplectic manifold $M=T^*\R^n\simeq\R^{2n}$. In that case, we
use the notation $(q,p)$, with $q\equiv(q^1,\ldots,q^n)$ and
$p\equiv(p_1,\ldots,p_n)$, for the canonical coordinates on $M$, and set
$\omega:=\sum_{j=1}^n\d q^j\wedge\d p_j$ for the canonical symplectic form. We always
assume that $f=\chi_1$ or that $f$ satisfies the hypotheses of Theorem \ref{IntCont}.

%--------------------------------------------------------------------------------------
\subsection{$\boldsymbol{\nabla H=g(H)}$}
%--------------------------------------------------------------------------------------

Suppose that there exists a function $g\equiv(g_1,\ldots,g_d)\in C^\infty(\R;\R^d)$
such that $\nabla H=g(H)$. Then $H$ and $\Phi$ satisfy Assumption \ref{AssCom} since
$\{g_j(H),H\}=0$ for each $j$. Furthermore, one has
$\Crit(H,\Phi)=(g\circ H)^{-1}(\{0\})$, and
$T_f=-\Phi\cdot(\nabla R_f)\big(g(H)\big)$ on $M\setminus\Crit(H,\Phi)$. We
distinguish various cases:
\begin{enumerate}
\renewcommand{\labelenumi}{{\normalfont (\Alph{enumi})}}

\item Suppose that $g$ is constant, \ie $g=v\in\R^d\setminus\{0\}$. Then 
$\Crit(H)=\Crit(H,\Phi)=\varnothing$, and we have the equality
$T_f=-\Phi\cdot(\nabla R_f)(v)$ on the whole of $M$.

Typical examples of functions $H$ and $\Phi$ fitting into this construction are
Friedrichs-type Hamiltonians and position functions. For illustration, we mention the
case (with $d=n$) of $H(q,p):=v\cdot p+V(q)$ and $\Phi(q,p):=q$ on $M:=\R^{2n}$, with
$v\in\R^n\setminus\{0\}$ and $V\in\cinf(\R^n;\R)$. In such a case, one has
$\nabla H=v$ and
$$
\textstyle
\varphi_t(q,p)=\big(vt+q,p-\int_0^t\d s\,(\nabla V)(vs+q)\big).
$$
Stark-type Hamiltonians and momentum functions also fit into the construction, \ie
$H(q,p):=h(p)+v\cdot q$ and $\Phi(q,p):=p$ on $M:=\R^{2n}$, with
$v\in\R^n\setminus\{0\}$ and $h\in \cinf(\R^n;\R)$. In such a case, one has
$\nabla H=-v$ and
$$
\textstyle
\varphi_t(q,p)=\big(q+\int_0^t\d s\,(\nabla h)(p-vs),p-vt\big).
$$
Note that these two examples are interesting since the Hamiltonians $H$ contain not
only a kinetic part, but also a potential perturbation.

\item Suppose that $\Phi$ has only one component ($d=1$), and assume that
$g(\lambda)=\lambda$ for all $\lambda\in\R$ (in the Hilbertian framework, one says in
such a case that $H$ is $\Phi$-homogeneous \cite{BG91}). Then
$\Crit(H,\Phi)=H^{-1}(\{0\})$ and we have the equality $T_f=-\Phi(\nabla R_f)(H)$ on
$M\setminus H^{-1}(\{0\})$. We present a general class of pairs $(H,\Phi)$ satisfying
these assumptions:

The Hamiltonian flow of the function $D(q,p):=q\cdot p$ on $\R^{2n}$ is given by
$\varphi^D_t(q,p)=(\e^tq,\e^{-t}p)$. So, $D$ is the generator of a dilations group on
$\R^{2n}$ (in the Hilbertian framework, the corresponding operator is the usual
generator of dilations on $\ltwo(\R^n)$, see \eg \cite[Sec.~1.2]{ABG}). Therefore,
the relation $\{D,H\}\propto H$ holds for a large class of homogeneous functions $H$
on $\R^{2n}$, due to Euler's homogeneous function theorem. Let us consider an
explicit situation. Take $\alpha>0$ and let $M$ be some open subset of
$(\R^n\setminus\{0\})\times\R^n$. Define on $M$ the function $\Phi:=\frac1\alpha D$
and the Hamiltonian $H(q,p):=h(p)+V(q)$, where $h\in\cinf(\R^n;\R)$ is positive
homogeneous of degree $\alpha$ and $V\in\cinf(\R^n\setminus\{0\};\R)$ is positive
homogeneous of degree $-\alpha$. Then one has $\nabla H\equiv\{\Phi,H\}=H$ on $M$,
and
\begin{align*}
\Crit(H)
&=\big\{(q,p)\in M\mid(\nabla h)(p)=(\nabla V)(q)=0\big\}\\
&\subset\big\{(q,p)\in M\mid p\cdot(\nabla h)(p)=q\cdot(\nabla V)(q)=0\big\}\\
&=\big\{(q,p)\in M\mid H(q,p)=0\big\}\\
&=\Crit(H,\Phi).
\end{align*}
Furthermore, if the functions $h$ and $V$ and the subset $M$ are well chosen, the
Hamiltonian vector field $X_H$ of $H$ is complete. For instance,
\begin{enumerate}
\item[(i)] If $V\equiv0$, then one can take $M=\R^{2n}$, and one has
$\varphi_t(q,p)=\big(q+t(\nabla h)(p),p\big)$ and
$$
\Crit(H)
=\big\{(q,p)\in M\mid(\nabla h)(p)=0\big\}
\subset\big\{(q,p)\in M\mid p\cdot(\nabla h)(p)=0\big\}
=\Crit(H,\Phi)
$$
(when $h(p)=\12|p|^2$ is the classical kinetic energy, one has
$\Crit(H)=\Crit(H,\Phi)=\R^n\times\{0\}$).

\item[(ii)] Let $K>0$. Then the Hamiltonian given by $H(q,p):=\12(|p|^2+K|q|^{-2})$
on $M:=\R^n\setminus\{0\}\times\R^n$ has a complete Hamiltonian vector field $X_H$.
To see it, we use the push-forward of $X_H$ by the diffeomorphism 
$\iota:\R^n\setminus\{0\}\times\R^n\to\R^n\setminus\{0\}\times\R^n$, 
$(q,p)\mapsto\big(q|q|^{-2},p\big)\equiv(r,p)$, namely,
$$
[\iota_*(X_H)](r,p)=\sum_j
\left(\big(|r|^2p_j-2(p\cdot r)r^j\big)\frac\partial{\partial r^j}\Big|_{(r,p)}
+Kr^j|r|^2\frac\partial{\partial p_j}\Big|_{(r,p)}\right).
$$
Then, we obtain that $\iota_*(X_H)$ is complete by using the criterion
\cite[Prop.~2.1.20]{AM78} with the proper function
$g:\R^n\setminus\{0\}\times\R^n\to[0,\infty)$ given by $g(r,p):=|p|^2+K|r|^2$. Since
$\iota$ is a diffeomorphism, this implies that $X_H$ is also complete (see
\cite[Lemma~1.6.4]{Jos05}).

\end{enumerate}

\item Many other examples with $\nabla H=g(H)$ can be obtained using homogeneous
Hamiltonians functions. For instance, consider $H(q,p):=q^2/q^1+q^1/q^2$ and
$\Phi(q,p):=p_1q^2+p_2q^1$ on $M:=(\R^2\setminus\{0\})\times\R^2$. Then one has
$\nabla H=H^2-4$, $\varphi_t(q,p)=\big(q,p-t\frac{\partial H}{\partial q}(q,p)\big)$
and
$$
\Crit(H)=\Crit(H,\Phi)=\big\{q\in\R^2\setminus\{0\}\mid q^1=\pm q^2\big\}\times\R^2.
$$
\end{enumerate}

%--------------------------------------------------------------------------------------
\subsection{$\boldsymbol{H=h(p)}$}
%--------------------------------------------------------------------------------------

Consider on $M:=\R^{2n}$ a purely kinetic Hamiltonian $H(q,p):=h(p)$ with
$h\in\cinf(\R^n;\R)$, and take the usual position functions $\Phi(q,p):=q$ with
$d=n$. Then $\varphi_t(q,p)=\big(q+t(\nabla h)(p),p\big)$, $\nabla H=\nabla h$, and
Assumption \ref{AssCom} is satisfied:
$$
\big\{\{\Phi_j,H\},H\big\}
=\big\{(\partial_jh)(p),h(p)\big\}
=0.
$$
In this example, we have $\Crit(H)=\Crit(H,\Phi)=\R^n\times(\nabla h)^{-1}(\{0\})$.

%--------------------------------------------------------------------------------------
\subsection{The assumption $\boldsymbol{\{\{\Phi_j,H\},H\}=0}\,$ as a differential
equation}\label{S-EqDiff}
%--------------------------------------------------------------------------------------

Consider on $M:=\R^{2n}$ an Hamiltonian function $H$ with partial derivatives
$H_{p_k}:=\partial H/\partial p_k$ and $H_{q^k}:=\partial H/\partial q^k$. Then,
finding the functions $\Phi_j$ of Assumption \ref{AssCom} amounts to solving for
$\Phi_0$ the second-order linear equation
$$
\big\{\{\Phi_0,H\},H\big\}
\equiv\bigg(\sum_{\ell=1}^n \big(H_{p_\ell}\partial_{q^\ell}
-H_{q^\ell}\partial_{p_\ell}\big)\bigg)^2\Phi_0
=0.
$$
As observed in Section \ref{Sec_Crit}, this is essentially equivalent (when $k$
independent first integrals $J_1\equiv H,J_2,\ldots,J_k$ are known) to find the
solutions $\Phi_0$ to
\begin{equation}\label{eq_diff}
\{\Phi_0,H\}
=\sum_{\ell=1}^n \big(H_{p_\ell}\partial_{q^\ell}
-H_{q^\ell}\partial_{p_\ell}\big)\Phi_0
=g(J_1,\ldots,J_k).
\end{equation}
The case $g\equiv1$ is sufficient, though trying to solve $\{\Phi_0,H\}=1$ can at
best provide solutions which are $\cinf$ outside the set $\Crit(H)$. A way to remove
these singularities could be to multiply the solutions by a function $g(H)$ that
vanishes and is infinitely flat on $\Crit(H)$. For instance, if $H\big(\Crit(H)\big)$
consists of a finite number of values $c_1,\ldots,c_s\in\R$, one could take 
$g(H)=\prod_{j=1}^s\e^{-(H-c_j)^{-2}}$. Another possibility is to restrict the study
to a submanifold $M'$ of $M$ (typically an open subset of the same dimension).
However, problems can arise as the same (induced) symplectic structure (or Poisson
bracket) must be used for the dynamic to remain unchanged; in particular, it must
checked that the Hamiltonian flow preserves $M'$.

\begin{enumerate}
\renewcommand{\labelenumi}{{\normalfont(\Alph{enumi})}}

\item Repulsive harmonic potential. In this example we first solve the equation
$\{\Phi_0,H\}=1$, and then correct the functions $\Phi_0$ to make them $\cinf$. So,
let us consider for $K\neq0$ the Hamiltonian $H(q,p):=\12\big(|p|^2-K^2|q|^2\big)$ on
$M:=\R^{2n}$. One can check that $\Crit(H)=\{0\}$ and that
$$
\textstyle
\varphi_t(q,p)=\big(\frac{Kq+p}{2K}\;\!\e^{Kt}+\frac{Kq-p}{2K}\;\!\e^{-Kt},
\frac{Kq+p}2\;\!\e^{Kt}-\frac{Kq-p}2\;\!\e^{-Kt}\big).
$$
For $j\in\{1,\ldots,n\}$, take $\Phi_j(q,p):=\frac1K\tanh^{-1}(Kq^j/p_j)$, where
$\tanh^{-1}(z)\equiv\12\ln\big|\frac{1+z}{1-z}\big|$ is $\cinf$ on
$\R\setminus\{\pm1\}$. Whenever $p_j=\pm Kq^j$, the $\Phi_j$ are not well-defined,
but outside these regions, they satisfy $\{\Phi_j,H\} =1$. It is possible in this
case to get rid of the singular regions. Indeed, the functions
$H_j(q,p):=\12\big(p_j^2-K^2(q^j)^2\big)$ are first integrals of the motion and the
singular regions correspond to the level sets $H_j^{-1}(\{0\})$. Therefore, the
functions $\Phi'_j:=\e^{-H_j^{-2}}\Phi_j$ are well-defined and satisfy Assumption
\ref{AssCom}:
$$
\big\{\{\Phi_j',H\},H\big\}=\big\{\e^{-H_j^{-2}},H\big\}=0.
$$
In this example, one has
$\{0\}=\Crit(H)\subsetneq\Crit(H,\Phi')=\bigcap_jH_j^{-1}(\{0\})$.

\item Simple pendulum. In this example we first consider the dynamics on a manifold
and then restrict it to an appropriate submanifold. For $K>0$, take
$H(q,p):=\12\big(p^2+K(1-\cos q)\big)$ on $M:=\R^2$. One has
$\Crit(H)=\pi\Z\times\{0\}$ (the values $q\in2\pi\Z$ correspond to minima, while
$q\in2\pi\Z+\pi$ correspond to inflexion points). Then, consider the open subset $M'$
of $M$ defined by the relation $H>K$, \ie
$M':=\big\{(q,p)\in\R^2\mid p^2/2-K\cos^2(q/2)>0\big\}$. One verifies easily that
$M'$ is preserved by the Hamiltonian flow, that $M'\cap\Crit(H)=\varnothing$ and that
$M'$ corresponds to the region where the values of $q$ along an orbit cover all of
$\R$. Define also
$$
\Phi(q,p):=\sqrt{\frac2{H(q,p)}}\;\!F\big(q/2\big|\sqrt{K/H(q,p)}\big)
\equiv\sqrt2\int_0^{q/2}\frac{\d\vartheta}{\sqrt{H(q,p)-K\sin^2(\vartheta)}}\;\!,
$$
where $F(\;\!\cdot\;\!|\;\!\cdot\;\!)$ denotes the incomplete elliptic integral of
the first kind. Then one verifies that the function $\Phi$ is well-defined on $M'$
and a direct calculation gives $\{\Phi,H\}(q,p)=p/|p|$ for each $(q,p)\in M'$. Now,
$p/|p|=1$ on one connected component of $M'$ and $p/|p|=-1$ on the other one. Thus
Assumption \ref{AssCom} is verified on $M'$ and $\Crit(H,\Phi)=\varnothing$.

\item Unbounded trajectories of central force systems. Once again, we first consider
the dynamics on a manifold and then restrict it to an appropriate submanifold. For
$K\in\R\setminus\{0\}$, take $H(q,p):=\12\big(|p|^2-K|q|^{-1}\big)$ on 
$M:=(\R^n\setminus\{0\})\times\R^n$, with $n>1$ if  $K>0$ and $n\geq1$ if $K<0$. One
has $\Crit(H)=\varnothing$.

When $K>0$ (and $n>1$), we must restrict our attention to the case where the
Hamiltonian function $H$ is positive (to avoid periodic orbits), and where at least
one of the two-dimensional angular momenta $L_{ij}(q,p):=q^ip_j-q^jp_i$ is nonzero
(to avoid collisions, \ie orbits whose flow is not defined for all $t\in\R$, see
\cite{OV06}). Therefore, the open set
$
M':=\big\{(q,p)\in M|H(p,q)>0,\,\sum_{i,j=1}^n|L_{ij}(q,p)|^2\neq0\big\}
$
is an appropriate submanifold of $M$ when $K>0$. 

Consider now the real valued functions on $M$ (resp. $M'$) when $K<0$ (resp. $K>0$
and $n>1$) given by
$$ 
\Phi_\pm(q,p):=\frac{p\cdot q}{2H(q,p)}
\mp\frac K{2\big(2H(q,p)\big)^{3/2}}\;\!\ln\Big(|q|\big(2H(q,p)+|p|^2\big)
\pm2\sqrt{2H(q,p)}\;\!p\cdot q\Big).
$$
Since $|p|^2<2H(q,p)$ (resp. $|p|^2>2H(q,p)$), then
\begin{align*}
\big(\sqrt{2H(q,p)}-|p|\big)^2>0
&\Longrightarrow 2H(q,p)|p|^2\pm2\sqrt{2H(q,p)}\;\!p\cdot\frac q{|q|}>0\\
&\hspace{-5pt}\iff|q|\big(2H(q,p)+|p|^2\big)\pm2\sqrt{2H(q,p)}\;\!p\cdot q>0.
\end{align*}
So, $\Phi_\pm$ are well-defined, and further calculations show that
$\{\Phi_\pm,H\}=1$ on $M$ (resp. $M'$). As before,
$\Crit(H)=\Crit(H,\Phi_\pm)=\varnothing$. Note that $\Phi_\pm(q,p)=p\cdot q/|p|^2$
when $K=0$, which is coherent with the canonical function $\Phi$ for the purely
kinetic Hamiltonian $H(q,p)=\12|p|^2$.

One can construct a more intuitive function $\Phi_0$ in terms of $\Phi_\pm$, namely,
$$
\Phi_0(q,p)
:=\12(\Phi_++\Phi_-)(q,p)
=\frac{p\cdot q}{2H(q,p)}-\frac K{2\big(2H(q,p)\big)^{3/2}}\;\!
\tanh^{-1}\bigg(\frac{2\sqrt{2H(q,p)}\;\!p\cdot q}{|q|\big(2H(q,p)+|p|^2\big)}\bigg),
$$
which also satisfies $\{\Phi_0,H\}=1$. Since the functions satisfying Assumption
\ref{AssCom} are linear in $t$, one can regard them as inverse functions for the
flow. The appearance of the inverse hyperbolic function $\tanh^{-1}$ in
$\Phi_0$ is related to the fact that unbounded trajectories of the central force
system given by $H>0$ are hyperbolas.  

\item Poincar\'e ball model. Consider $B_1:=\big\{q\in\R^n\mid|q|<1\big\}$ endowed
with the Riemannian metric $g$ given by
$$
g_q(X_q,Y_q):=\frac4{(1-|q|^2)^2}\;\!(X_q\cdot Y_q),\quad
q\in B_1,~X_q,Y_q\in T_qB_1\simeq\R^n.
$$
Let $M:=T^*B_1\simeq\big\{(q,p)\in B_1\times\R^n\big\}$ be the cotangent bundle on
$B_1$ with symplectic form $\omega:=\sum_{j=1}^n\d q^j\wedge\d p_j$, and let
$$
H:M\to\R,\quad(q,p)\mapsto\12\sum_{j,k=1}^ng^{jk}(q)p_jp_k
={\textstyle \frac18}|p|^2\big(1-|q|^2\big)^2
$$
be the kinetic energy Hamiltonian. It is known that the integral curves of the vector
field $X_H$ correspond to the geodesics curves of $(B_1,g)$ (see
\cite[Thm.~1.6.3]{Jos05} or \cite[Sec.~6.4]{CC05}). Since, $(B_1,g)$ is geodesically
complete (see Proposition 3.5 and Exercice 6.5 of \cite{Lee97}), this implies that
$X_H$ is complete. There remains only to find a function $\Phi$ satisfying Assumption
\ref{AssCom} in order to apply the theory.

Some calculations using spherical-type coordinates suggest the function
$$
\Phi:M\to\R,\quad(q,p)\mapsto\e^{-1/H(q,p)}
\tanh^{-1}\left(\frac{(p\cdot q)(1-|q|^2)}{\sqrt{2H(q,p)}(1+|q|^2)}\right).
$$
Indeed, since
$$
\left|\frac{(p\cdot q)(1-|q|^2)}{\sqrt{2H(q,p)}(1+|q|^2)}\right|
=\left|\frac{2(p\cdot q)}{|p|(1+|q|^2)}\right|
\leq\frac{2|q|}{1+|q|^2}
<1,
$$
the function $\Phi$ is well-defined. Furthermore, direct calculations show that
$\Phi$ is $\cinf$ and that $\{\Phi,H\}=\e^{-1/H}\sqrt{2H}$. Therefore, Assumption
\ref{AssCom} is verified and one has $\Crit(H)=\Crit(H,\Phi)=B_1\times\{0\}$.

In one dimension, $q(t):=\tanh(t)$ is (up to speed and direction) the only geodesic
curve, and
$$
\Phi(q,p)
=\e^{-1/H(q,p)}\tanh^{-1}\left(\frac{2pq}{|p|(1+q^2)}\right)
=2\e^{-1/H(q,p)}\frac p{|p|}\tanh^{-1}(q).
$$
So, apart from the smoothing factor $2\e^{-1/H}$, our $\Phi$ coincides in one
dimension with the inverse function of the flow.
\end{enumerate}

%--------------------------------------------------------------------------------------
\subsection{Passing to a covering manifold}
%--------------------------------------------------------------------------------------

In this subsection we briefly discuss a way of avoiding the obstruction of periodic
orbits: Given $M$ a symplectic manifold with symplectic form $\omega$ and Hamiltonian
$H$, we let $\pi:\widetilde M\to M\setminus\Crit(H)$ be $\cinf$-covering manifold. In
order to preserve the dynamics, we endow the manifold $\widetilde M$ with the
pullback $\widetilde\omega:=\pi^*\omega$ of the symplectic form $\omega$ and with the
pullback $\widetilde H:=\pi^*H$ of the Hamiltonian $H$.\footnote{If one wants to
consider only a Poisson manifold $M$, a Poisson structure can also be defined on
$\widetilde{M}$ given that $\pi$ is $\cinf$. Indeed, for
$U\subset M\setminus\Crit(H)$ a sufficiently small open set (\ie such that
$\pi^{-1}(U)$ is a disjoint union of diffeomorphic copies), connected components of
$\pi^{-1}(U)$ are diffeomorphic to $U$ and the Poisson structure can be induced by
this diffeomorphism.}

Here are two simple examples of finite-dimensional symplectic covering manifolds.

\begin{enumerate}
\renewcommand{\labelenumi}{{\normalfont(\Alph{enumi})}}

\item Consider on the sphere $M:=\mathbb S^2$ (as seen in $\R^3$ and with its
standard symplectic structure) the Hamiltonian $H$ given by the projection onto the
$z$-coordinate. Outside the $2$ polar critical points, all the orbits are periodic:
the flow corresponds to rotations around the $z$-axis. In this case, one can use the
covering of $\mathbb S^2\setminus\{(0,0,\pm1)\}$ given by
$\widetilde M:=\big\{(\vartheta,z)\mid\vartheta\in\R,~z\in(-1,1)\}$ and the covering
map
\begin{align*}
\pi:\widetilde M\to M\setminus\Crit(H)
\equiv\mathbb S^2\setminus\{(0,0,\pm1)\},\quad
(\vartheta,z)\mapsto\big(\sqrt{1-z^2}\;\!\cos(\vartheta),
\sqrt{1-z^2}\;\!\sin(\vartheta),z\big).
\end{align*}
Consequently, $\widetilde H:\widetilde M\to(-1,1)$ is the projection onto the
$z$-coordinate and $\widetilde\omega=\d\vartheta\wedge\d z$. One can also check that
$\varphi_t(\vartheta,z):=(\vartheta+t,z)$ is the flow of $\widetilde H$ and that
$\big\{\Phi,\widetilde H\big\}=1$ for $\Phi(\vartheta,z):=\vartheta$. So, Assumption
\ref{AssCom} is verified on $\widetilde M$ and
$\Crit(\widetilde H)=\Crit(\widetilde H,\Phi)=\varnothing$.

\item Harmonic oscillator. Consider on $M:=\R^{2n}$ (with its standard symplectic
structure) the Hamiltonian given by $H(q,p):=\12\big(|p|^2+K^2|q|^2\big)$, where
$K\in\R\setminus\{0\}$. Define
$\widetilde M:=\big\{(r,\vartheta)\mid r\in(0,\infty)^n,\,\vartheta\in\R^n\big\}$ and
$\pi:\widetilde M\to M\setminus\Crit(H)\equiv\R^{2n}\setminus\{0\}$, with
$$
\pi(r,\vartheta)
:=\big(K^{-1}r_1\cos(\vartheta_1),\ldots,K^{-1}r_n\cos(\vartheta_n),
r_1\sin(\vartheta_1),\ldots,r_n\sin(\vartheta_n)\big).
$$
Then $\widetilde H(r,\vartheta)=\12|r|^2$,
$\widetilde\omega=K^{-1}\sum_{j=1}^nr_j\;\!\d r_j\wedge \d\theta_j$, and
$\varphi_t(r,\vartheta)=(r,\vartheta-Kt)$ is the flow of $\widetilde H$. Furthermore,
one has $\big\{\Phi_j,\widetilde H\big\}=-K$ for each function
$\Phi_j(r,\vartheta):=\vartheta_j$. Therefore, Assumption \ref{AssCom} is verified on
$\widetilde M$ with $\Phi\equiv(\Phi_1,\ldots,\Phi_n)$ and
$\Crit(\widetilde H)=\Crit(\widetilde H,\Phi)=\varnothing$.
\end{enumerate}

%--------------------------------------------------------------------------------------
\subsection{Infinite dimensional Hamiltonian systems}
%--------------------------------------------------------------------------------------

%--------------------------------------------------------------------------------------
\subsubsection{Classical systems}
%--------------------------------------------------------------------------------------

In the following examples, the infinite dimensional manifold $M$ is either
$\ltwo(\R)$ or $\ltwo(\R)\oplus\ltwo(\R)$ (equivalence classes of real valued square
integrable functions)\footnote{In the case of the wave and the Schr\"odinger
equations below, one can easily extend the results to the situation where $\ltwo(\R)$
is replaced by $\ltwo(\R^n)$. We restrict ourselves to the case $n=1$ for the sake of
notational simplicity.}. The atlas of $M$ consists in only one chart, the tangent
space $T_uM$ at a point $u\in M$ is isomorphic to $M$, and the Riemannian metric on
$M$ is flat (\ie independent of the base point in $M$) and given by the usual scalar
product $\langle\;\!\cdot\;\!,\;\!\cdot\;\!\rangle$ on $\ltwo(\R)$ or
$\ltwo(\R)\oplus\ltwo(\R)$.

To define the symplectic form on $M$ in terms of the metric
$\langle\;\!\cdot\;\!,\;\!\cdot\;\!\rangle$ we let $\H^s$, $s\in\R$, denote the real
Sobolev space $\H^s(\R)$ or $\H^s(\R)\oplus\H^s(\R)$ (see \cite[Sec.~4.1]{ABG} for
the definition in the complex case) and we let $\S$ denote the real Schwartz space
$\S(\R)$ or $\S(\R)\oplus\S(\R)$. Then we consider an operator $J:\S\to\S$ (which can
be interpreted by continuity as an endomorphism of the tangent spaces $T_uM\simeq M$)
satisfying the following:
\begin{enumerate}
\item[(i)] There exists a number $d_J\geq0$, called the order of $J$, such that for
each $s\in\R$ the operator $J$ extends to an isomorphism $\H^s\to\H^{s-d_J}$ (which
we denote by the same symbol).
\item[(ii)] $J$ is antisymmetric on $\S$, \ie
$\langle Jf,g\rangle=-\langle f,Jg\rangle$ for all $f,g\in\S$.
\end{enumerate}
It is known \cite[Lemma~1.1]{Kuk93} that the operator
$\bar J:=-J^{-1}:M\to\H^{d_J}$ (of order $-d_J$) is bounded and anti-selfadjoint in
$M$. In consequence, for each $s\geq0$ the map $\omega:\H^s\times\H^s\to\R$ given by
$$
\omega(f,g):=-\big\langle\bar Jf,g\big\rangle
$$
defines a symplectic form on $\H^s$.

The functions on the phase space (such as $H$ or $\Phi_j$) are infinitely Fr\'echet
differentiable mappings from $\O_{s_H}$ (a subset of $\H^{s_H}$ for some $s_H\geq0$)
to $\R$, \ie elements of $\cinf(\O_{s_H};\R)$. The Hamiltonian function
$H\in\cinf(\O_{s_H};\R)$ is defined as follows: for some $h\in\cinf(\R^{k+1};\R)$ (or
$h\in\cinf(\R^{2(k+1)};\R)$ if $M=\ltwo(\R)\oplus\ltwo(\R)$), one has for each
$u\in\O_{s_H}$
$$
H(u):=\int_\R\d x\,h(u_0,u_1,\ldots,u_{s_H}),
$$
where $u_j:=\frac{\d^ju}{\d x^j}$\;\!. Since $H\in\cinf(\O_{s_H};\R)$, the
differential of $H$ at $u\in\O_{s_H}$ on a tangent vector
$f\in\S\subset M \simeq T_u M$ is given by
$$
\d H_u(f)
=\lim_{t\to0}\frac1t\big[H(u+tf)-H(u)\big]
=\int_\R\d x\;\!\sum_{j=0}^{s_H}\frac{\partial h}{\partial u_j}\frac{\d^jf}{\d x^j}
=\sum_{j=0}^{s_H}\int_\R\d x\,(-1)^jf\;\!\frac{\d^j}{\d x^j}
\frac{\partial h}{\partial u_j}\;\!,
$$
where the second equality is obtained using integrations by parts (with vanishing
boundary contributions). The (Riemannian) gradient vector field $\grad H$ associated
to the linear functional $\d H$ satisfies by definition
$\big\langle(\grad H)(u),f\big\rangle=\d H_u(f)$ for all $u\in\O_{s_H}$ and $f\in\S$
(here $(\grad H)(u)$ \apriori only belongs to the topological dual $\S^*$ of $\S$,
which means that $\langle\;\!\cdot\;\!,\;\!\cdot\;\!\rangle$ denotes \apriori the
duality map between $\S^*$ and $\S$). So, $(\grad H)(u)$ is given by
\begin{equation}\label{gradienf}
(\grad H)(u)
=\sum_{j=0}^{s_H}(-1)^j\frac{\d^j}{\d x^j}\frac{\partial h}{\partial u_j}\;\!.
\end{equation}
Then, the Hamiltonian vector field $X_H$ is the map $\O_{s_H}\to\S^*$ satisfying
$$
\big\langle\bar Jf,X_H(u)\big\rangle
=-\omega\big(f, X_H(u)\big)
=\d H_u(f)
=\big\langle f,(\grad H)(u)\big\rangle
$$
for all $u\in\O_{s_H}$ and $f\in\S$. Since $\bar J$ is anti-selfadjoint, this implies
that $\bar JX_H(u)=-(\grad H)(u)$ in $\S^*$, which is equivalent to
$X_H(u)=J(\grad H)(u)$ in $\S^*$. So, the equation of motion with Hamiltonian $H$ has
the form $\ddt\;\!u=J(\grad H)(u)$, and
$\{\Phi,H\}=\d\Phi(X_H)=\big\langle\grad\Phi,J(\grad H)\big\rangle$ for all functions
$\Phi,H\in\cinf(\O_{s_H};\R)$ with appropriate gradient.

Before passing to concrete examples, we refer to \cite{Kat75} for standard results on
the local existence in time of Hamiltonian flows (global existence is specific to the
system considered).

\begin{enumerate}
\renewcommand{\labelenumi}{{\normalfont(\Alph{enumi})}}

\item The wave equation. We refer to \cite[Ex.~5.5.1]{AM78},
\cite[Ex.~8.1.12]{AMR88}, \cite[Sec.~2.1]{CM74} and \cite[Sec.~X.13]{RS75} for a
description of the model. The existence of the flow for all times depends on the
nonlinear term in the Hamiltonian (see for instance \cite[Thm.~X.74]{RS75} and the
corollary that follows).

In this example, the scale $\{\H^s\}_{s\geq0}$ is given by
$\H^s:=\H^s(\R)\oplus\H^s(\R)$. The metric on $M:=\ltwo(\R)\oplus\ltwo(\R)$ is given
for each $(p,q),(\widetilde p,\widetilde q)\in M$ by
$
\big\langle(p,q),(\widetilde p,\widetilde q)\big\rangle
:=\int_\R\d x\,(p\widetilde p+q\widetilde q)
$,
and the operator $J$ is given by
$$
J:M\to M,\quad(p,q)\mapsto(-q,p).
$$
It is an isomorphism of degree $0$ with $\bar J=J$. Given $m\geq0$ and
$F\in\cinf(\R;\R)$, one can find a subset $\O_1\subset\H^1$ (depending on $F$) such
that the Hamiltonian function
$$
H:\O_1\to\R,\quad(p,q)\mapsto\int_\R\d x\,h(p,q,\partial_xq)
\equiv\12\int_\R\d x\,\big\{p^2+(\partial_xq)^2+m^2q^2+2F(q)\big\},
$$
is well-defined and $\cinf$. In fact, we assume that $\O_1$ is chosen such that (i)
all the functions on the phase space appearing below are elements of
$\cinf(\O_1;\R)$, and (ii) integrations by parts involving these functions come
vanishing boundary contributions. Then one checks that
$(\grad H)(p,q)=\big(p,m^2q+F'(q)-\partial_x^2q\big)$ due to \eqref{gradienf}, and
that $X_H(p,q)$ is trivial if and only if $p=0$ and $m^2q+F'(q)-\partial_x^2q=0$. The
constraint on $q$ depends on the choice of $F$. For example, when $F(q)=0,q$ or
$q^2$, the solution $q$ of the differential equation does not decay as
$|x|\to\infty$. In consequence, the corresponding pairs $(p,q)$ cannot belong to $M$,
and $\Crit(H)=\{(0,0)\}$. The equation of motion
\begin{equation}\label{eq_motion}
\ddt\;\!(p,q)=J(\grad H)(p,q)
\end{equation}
coincides with the usual the wave equation since the combination of
$\ddt p=\partial_x^2q-m^2q-F'(q)$ and $\ddt q=p$ gives
$$
\frac{\d^2}{\d t^2}\;\!q=\partial_x^2q-m^2q-F'(q).
$$
When $m\neq0$, this equation is called the Klein-Gordon equation, and $F$ is usually
assumed to be a nonlinear term of the form $F(q)=q^\lambda$ for some $\lambda\in\R$.
A first relevant observation is that the function $C_0\in\cinf(\O_1;\R)$ given by
$C_0(p,q):=\int_\R\d x\,p(\partial_x q)$ is a first integral of the motion.
Furthermore, the function $\Phi_0\in\cinf(\O_1;\R)$ given by
$\Phi_0(p,q):=\int_\R\d x\,\id_\R h(p,q,\partial_x q)$ has gradient
$
(\grad \Phi_0)(p,q)
=\big(\id_\R p,\id_\R m^2q+\id_\R F'(q)-\partial_x(\id_\R\partial_xq)\big).
$
Therefore,
$$
\{\Phi_0,H\}(p,q)
=\big\langle(\grad\Phi_0)(p,q),J(\grad H)(p,q)\big\rangle
=\int_\R\d x\,p\;\!\big\{\id_\R\partial_x^2q-\partial_x(\id_\R\partial_xq)\big\}
=-C_0(p,q),
$$
and $\Phi_0$ satisfies Assumption \ref{AssCom}. Here, we clearly have
$$
\textstyle
\Crit(H,\Phi_0)
=C_0^{-1}(\{0\})
=\big\{(p,q)\in\O_1\mid\int_\R p(\partial_xq)\,\d x=0\big\}
\supsetneq\{(0,0)\}
=\Crit(H).
$$ 
If we assume further that $F\equiv0$, then the equation of motion \eqref{eq_motion}
is linear. Therefore any pair $(\partial_x^jp,\partial_x^j q)$, $j\geq1$, with
$(p,q)$ a solution of \eqref{eq_motion}, also satisfies \eqref{eq_motion}.
Consequently, if the subsets $\O_j\subset\H^j$ have properties similar to the ones of
$\O_1$, then the functions $C_j\in\cinf(\O_{j+1};\R)$ and $H_j\in\cinf(\O_{j+1};\R)$
given by
$
C_j(p,q):=\int_\R\d x\,\big(\partial_x^jp\big)\big(\partial_x^{j+1}q\big)
$
and
$
H_j(p,q):=\int_\R\d x\,h\big(\partial_x^jp,\partial_x^jq,\partial_x^{j+1}q\big)
$
are first integrals of the motion. Accordingly, one deduces that the functions
$\Phi_j\in\cinf(\O_{j+1};\R)$ given by
$
\Phi_j(p,q)
:=\int_\R\d x\,\id_\R h\big(\partial_x^jp,\partial_x^jq,\partial_x^{j+1}q\big)
$
satisfy $\{\Phi_j,H\}=-C_j$ on $\O_{j+1}$. So, if $F\equiv0$,
there is an infinite family of functions $\Phi_j$ satisfying Assumption \ref{AssCom},
and one has again $\Crit(H,\Phi_j)\supsetneq\Crit(H)$, with
$\partial_x^j:\Crit(H,\Phi_j)\to\Crit(H,\Phi_0)$ an isomorphism.

Finally, when $F\equiv0$ and $m=0$ one can check that the function
$\widetilde\Phi_0\in\cinf(\O_1;\R)$ given by
$\widetilde\Phi_0(p,q):=\int_\R\d x\,\id_\R p(\partial_xq)$ has gradient
$
\big(\grad\widetilde\Phi_0\big)(p,q)=(\id_\R\partial_xq,-\id_\R\partial_xp-p).
$
Then,
$$
\big\{\widetilde\Phi_0,H\big\}(p,q)
=\int_\R\d x\,
\big(\id_\R(\partial_xq)(\partial_x^2q)-\id_\R p\;\!\partial_xp-p^2\big)
=-\12\int_\R\d x\,\big((\partial_x q)^2+p^2\big)
=-H(p,q),
$$
where the third equality is obtained using integrations by parts (with vanishing
boundary contributions). Thus $\widetilde\Phi_0$ satisfies Assumption \ref{AssCom}.
Furthermore, since $\{\widetilde\Phi_0,H\}(p,q)= 0$ implies
$\int_\R\d x\,\big\{(\partial_x q)^2+p^2\big\}=0$, one has
$\Crit(H,\widetilde\Phi_0)=\Crit(H)=\{(0,0)\}$. As before, any derivative of a
solution of the equation of motion is still a solution of the equation of motion. So,
it can be checked that the functions $\widetilde\Phi_j\in\cinf(\O_{j+1};\R)$ given by
$
\widetilde\Phi_j(p,q)
:=\int_\R\d x\,\id_\R\big(\partial_x^j p\big)\big(\partial_x^{j+1}q\big)
$
satisfy $\{\widetilde\Phi_j,H\}=-H_j$ on $\O_{j+1}$. Therefore, one has once again
$\Crit(H,\widetilde\Phi_j)=\Crit(H)=\{(0,0)\}$ and the $\widetilde\Phi_j$'s
constitutes a second infinite family of functions satisfying Assumption \ref{AssCom}.

\item The nonlinear Schr{\"o}dinger equation. We refer to
\cite[Ex.~1.3,~p.~3\,\&\,5]{Kuk93} for a description of the model. The existence of
the flow for all times depends on the nonlinear term in the Hamiltonian (see for
instance \cite[Sec.~I.2]{Bou99} and \cite[Sec.~3.2.2-3.2.3]{Sul99}).

The setting is the same as that of the previous example, except that the Hamiltonian
function $H\in\cinf(\O_1;\R)$ is given by
$$
H(p,q):=\12\int_\R\d x\,\big\{(\partial_xp)^2+(\partial_xq)^2+V\cdot(p^2+q^2)
+F(p^2+q^2)\big\},
$$
where $V,F\in\cinf(\R;\R)$. Using \eqref{gradienf}, one checks that the gradient of $H$
at $(p,q)\in\O_1$ is
$$
(\grad H)(p,q)
=\big(-\partial_x^2p+Vp+pF'(p^2+q^2),-\partial_x^2q+Vq+qF'(p^2 +q^2)\big).
$$
So, the equation of motion $\ddt(p,q)=J({\rm grad}H)(p,q)$ is equivalent to the
nonlinear Schr\"odinger equation
\begin{equation}\label{NLS}
\ddt\;\!u=i\big(-\partial_x^2u+Vu+uF'(|u|^2)\big),
\end{equation}
with $u:=p+iq$. Without additional assumptions on $F$ or $V$, it is hardly possible
to determine the set $\Crit(H)$ of functions $u$ for which the r.h.s. of \eqref{NLS}
vanishes. However, it is known that in general $\Crit(H)$ is not trivial, as in the
case of elliptic stationary nonlinear Schr\"odinger equations (see Theorem 1.1 and
Proposition 1.1 of \cite{BL90}).

Now, assume that $V\equiv F\equiv0$ and for each $j\geq1$ let $\O_j\subset\H^j$ be a
subset having properties similar to the ones of $\O_1$. Then the functions
$H_j\in\cinf(\O_j;\R)$ and $C_j\in\cinf(\O_{j+1};\R)$ given by
$
H_j(p,q)
:=\12\int_\R\d x\,\big\{(\partial_x^jq)^2+(\partial_x^jp)^2\big\}
\equiv\int_\R\d x\,h_j(p,q)
$
and
$
C_j(p,q)
:=\int_\R\d x\,\big\{(\partial_x^jq)(\partial_x^{j+1}p)
-(\partial_x^{j+1}q)(\partial_x^j p)\big\}
\equiv\int_\R\d x\,c_j(p,q)
$
are first integrals of the motion. Furthermore, the functions
$\Phi_j\in\cinf(\O_j;\R)$ and $\widetilde\Phi_j\in\cinf(\O_{j+1};\R)$ given by
$\Phi_j(p,q):=\int_\R\d x\, \id_\R h_j(p,q)$ and
$\widetilde\Phi_j(p,q):=\int_\R\d x\,\id_\R c_j(p,q)$ satisfy $\{\Phi_j,H\}=C_j$ and 
$\{\widetilde\Phi_j,H\}=4H_{j+1}$ on $\O_{j+1}$. So, the $\Phi_j$'s and the
$\widetilde\Phi_j$'s constitute two infinite families of functions satisfying
Assumption \ref{AssCom}. Note that the sets
$
\Crit(H,\Phi_j)
=C_j^{-1}(\{0\})
=\big\{(p,q)\in\O_{j+1}\mid\int_\R\d x\,(\partial_x^j q)(\partial_x^{j+1}p)= 0\big\}
$
(with isomorphisms $\partial_x^j:\Crit(H,\Phi_j)\to\Crit(H,\Phi_0)$) are rather
large, whereas $\Crit(H,\widetilde\Phi_j)=\Crit(H)=\{(0,0)\}$.

Some of the above functions still work when $V$ and $F$ are not trivial. For
instance, the identity $\{\Phi_0,H\}=C_0$ on $\O_1$ remains valid for all $V$ and
$F$. Furthermore, if $V={\rm Const.}$, then $\{C_0,H\}=0$ on $\O_1$. Consequently,
$\Phi_0$ satisfies Assumption \ref{AssCom} for all $F$ and for $V={\rm Const.}$, and
one has $\Crit(H,\Phi_0)\supsetneq\Crit(H)$. This last example is interesting since
it applies to a large class of nonlinear Schr{\"o}dinger equations.

\item The Korteweg-de Vries equation. Among many other possible references, we
mention \cite[Ex.~5.5.7]{AM78} and \cite[Ex.~1.4,~p.~3\,\&\,5]{Kuk93}. For the global
existence of the flow, we refer the reader to \cite[Sec.~1]{CKSTT03} and references
therein.

In this example, the scale $\{\H^s\}_{s\geq0}$ is given by $\H^s:=\H^s(\R)$ and the
sets $\O_j$, $j\in\N$, are appropriate subsets of $\H^j$. The Hamiltonian function
$H\in\cinf(\O_1;\R)$ is given by
$$
H(u):=\int_\R\d x\,\big( \12(\partial_xu)^2+u^3 \big),
$$
and the isomorphism $J:=\partial_x$ is of order $1$.

The gradient of $H$ at $u\in\O_1$ is $-\partial_x^2u+3u^2$. So, the elements of
$\Crit(H)$ are functions $u$ satisfying $-\partial_x^2u+3u^2=0$; these are
Weierstrass $\wp$-functions \cite[Sec.~134.F]{Kiy87}, that is, functions with many
singularities and no decay at infinity. Thus, $\Crit(H)=\{0\}$. Furthermore, the
equation of motion $\ddt u=J(\grad H)(u)$ coincides with the KdV equation
$\ddt u=\partial_x\big(-\partial_x^2 u + 3u^2\big)$.

There exists an infinite number of first integrals of the motion with polynomial
density, that is, of the form $H_j:=\int_\R\d x\,h_j$, where $h_j$ is a finite
polynomial in $u$ and its derivatives (see \cite[Sec.~3]{MGK68}). For example, when
$h_1(u)=u$, $h_2(u)=u^2$, $h_3(u)=\12(\partial_xu)^2+u^3$, or 
$h_4(u)=(\partial_x^2u)^2+10u(\partial_xu)^2+5u^4$. So, let $\Phi_0\in\cinf(\O_0;\R)$
be given by $\Phi_0(u):=\int_\R\d x\,\id_\R u$. Then the gradient of $\Phi_0$ at $u$
is $\id_\R$, and $\{\Phi_0,H\}=-3H_2$ on $\O_1$. Since $H_2$ is a first integral of
the motion, this implies that $\Phi_0$ satisfies Assumption \ref{AssCom}.
Furthermore, the fact that $H_2(u)=\|u\|_{\ltwo(\R)}$ implies that
$\Crit(H,\Phi_0)=\{0\}=\Crit(H)$.

Looking for others $\Phi$ of the form
$\Phi(u)=\int_\R\d x\,g(x)\;\!G(u,\partial_xu,\ldots,\partial_x^ku)$, with $G$ a
polynomial and $g$ a $\cinf$ function, is unnecessay. Indeed, both $\{\Phi,H\}$ and
$\Upsilon(t):=\Phi-t\{\Phi,H\}$ are first integrals of the motion with density
$\cinf$ in $x$ and polynomial in $u$ and its derivatives (and $t$-linear in the case
of $\Upsilon$). Thus, we know from \cite[Thm.~1~\&~Rem.~3]{SW97} that they are
completely characterised, up to the usual equivalence of conservation laws
\cite[Sec.~4.3]{Olv93}. Therefore, the functions $\Phi$ are also completely
characterised. Note however, that it is not excluded that functions $\Phi$ with an
integrand $G$ involving fractional derivatives, an infinite number of derivatives, or
of class $\cinf$ might work. Non-polynomial conserved densities are known to exist in
the periodic case (see \cite[Sec.~5]{MGK68}). 

\end{enumerate}

%--------------------------------------------------------------------------------------
\subsubsection{Quantum systems} \label{qsys}
%--------------------------------------------------------------------------------------

Let $\H$ be a complex Hilbert space, with scalar product
$\langle\;\!\cdot\;\!,\;\!\cdot\;\!\rangle$ antilinear in the left entry. Define on
$\H$ the usual quantum-mechanical symplectic form
$$
\omega:\H\times\H\to\R,\quad(\psi_1,\psi_2)\mapsto2\im\langle\psi_1,\psi_2\rangle.
$$
The pair $(\H,\omega)$ has the structure of an (infinite-dimensional) symplectic
vector space. Now, define for any bounded selfadjoint operator $H_{\rm op}\in\B(\H)$
the expectation value Hamiltonian function
$$
H:\H\to\R,\quad
\psi\mapsto\langle H_{\rm op}\rangle(\psi)
:=\langle\psi,H_{\rm op}\psi\rangle.
$$
Then, it is known \cite[Cor.~2.5.2]{MR99} that the vector field and the flow
associated to $H$ are $X_H=-iH_{\rm op}$ and $\varphi_t(\psi)=\e^{-itH_{\rm op}}\psi$.
Therefore, the Poisson bracket of two such Hamiltonian functions $H,K$ satisfies for
each $\psi\in\H$
$$
\{K,H\}(\psi)
=\omega\big(X_K(\psi),X_H(\psi)\big)
=-\omega(K_{\rm op}\psi,H_{\rm op}\psi)
=\big\langle\psi,i[K_{\rm op},H_{\rm op}]\psi\big\rangle.
$$
So, in this framework, verifying Assumption \ref{AssCom} amounts to finding
Hamiltonian functions $H\equiv\langle H_{\rm op}\rangle$ and 
$\Phi_j\equiv\langle(\Phi_j)_{\rm op}\rangle$ satisfying the commutation relation
\begin{equation}\label{doubleC}
\big[[(\Phi_j)_{\rm op},H_{\rm op}],H_{\rm op}\big]=0.
\end{equation}
In concrete examples, the operators $H_{\rm op}$ and $(\Phi_j)_{\rm op}$ are usually
unbounded. Therefore, the preceding calculations can only be justified (using the
theory of sesquilinear forms) on subspaces of $\H$ where all the operators are
well-defined. We do not present here the whole theory since much of it, examples
included, is similar to that of \cite{RT10}. We prefer to present a new example
inspired by \cite{GG05}, where all the calculations can be easily justified.

Let $U$ be an isometry in $\H$ admitting a number operator, that is, a selfadjoint
operator $N$ such that $UNU^*=N-1$. Define on $\H$ the bounded selfadjoint operators
$$
\textstyle
\Delta:=\re(U)\equiv\frac12(U+U^*)\qquad\hbox{and}\qquad
S:=\im(U)\equiv\frac1{2i}(U-U^*).
$$
Then we know from \cite[Sec.~3.1]{GG05} that any polynomial in $U$ and $U^*$ leaves
invariant the domain $\dom(N)\subset\H$ of $N$. In particular, the operator
$$
\textstyle
A_0:=\frac12(SN+NS),\quad\dom(A_0):=\dom(N),
$$
is well-defined and symmetric. In fact, it is shown that $A_0$ admits a selfadjoint
extension $A$ with domain $\dom(A)=\dom(NS)$. Furthermore, one has on $\dom(N)$ the
identity $i[A,\Delta]=\Delta^2-1$. So, if we define the Hamiltonian functions
$$
H:\H\to\R,\quad\psi\mapsto\langle\Delta\rangle(\psi)\qquad\hbox{and}\qquad
\Phi:\dom(N)\to\R,\quad\psi\mapsto\langle A\rangle(\psi),
$$
we obtain for each $\psi\in\dom(N)$
$$
(\nabla H)(\psi)
=\{\Phi,H\}(\psi)
=\langle i[A,H]\rangle(\psi)
=\langle\Delta^2-1\rangle(\psi),
$$
and Assumption \ref{AssCom} is verified for each $\psi\in\dom(N)$:
$$
\big\{\{\Phi,H\},H\big\}(\psi)
=\omega\big(X_{\langle\Delta^2-1\rangle}(\psi),X_{\langle\Delta\rangle}(\psi)\big)
=\big\langle i[\Delta^2-1,\Delta]\big\rangle(\psi)
=0.
$$
Now, since the spectrum of $\Delta$ is $[-1,1]$, the operator $1-\Delta^2$ is
positive, so we have the equivalences
$$
\langle\Delta^2-1\rangle(\psi)=0~~\Longleftrightarrow~~
\big\|(1-\Delta^2)^{1/2}\psi\big\|^2=0~~\Longleftrightarrow~~
\psi\in E^\Delta(\{\pm1\}).
$$
Thus,
$$
\Crit(H,\Phi)
\equiv(\nabla H)^{-1}(\{0\})
=\big\{\psi\in\dom(N)\mid\langle\Delta^2-1\rangle(\psi)=0\big\}
=\dom(N)\cap E^\Delta(\{\pm1\}).
$$
On the other hand, the elements $\psi\in\Crit(H)$ satisfy the condition
$$
0=X_H(\psi)=-i\Delta\psi~~\Longleftrightarrow~~\psi\in E^\Delta(\{0\}).
$$
This implies that $\Crit(H)=\{0\}$, since the spectrum of $\Delta$ is purely
absolutely continuous outside the points $\pm1$ \cite[Prop.~3.2]{GG05}. Finally, the
function $T_f$ is given by
$$
T_f=-\langle A\rangle\cdot(\nabla R_f)\big(\langle\Delta^2-1\rangle\big)
$$
on $\dom(N)\setminus\Crit(H,\Phi)$.

Typical examples of operators $\Delta$ and $N$ of the preceding type are Laplacians
and number operators on trees or complete Fock spaces (see \cite{GG05} for details).

%--------------------------------------------------------------------------------------
\section*{Acknowledgements}
%--------------------------------------------------------------------------------------

Part of this work was done while R.T.d.A was visiting the Max Planck Institute for
Mathematics in Bonn. He would like to thank Professor Dr. Don Zagier for his kind
hospitality. R.T.d.A also thanks Professor M. Musso for a useful conversation on the
stationary nonlinear Schr\"odinger equation.

%--------------------------------------------------------------------------------------
%\bibliography{../bibliographie/bibliographie}
%--------------------------------------------------------------------------------------

\end{document}